\documentclass[a4paper,onecolumn,unpublished]{quantumarticle}
\pdfoutput=1

\usepackage[utf8]{inputenc}
\usepackage[english]{babel}
\usepackage[T1]{fontenc}
\usepackage{amsmath}
\usepackage{mathrsfs}       
\usepackage{mathtools}
\usepackage{amsfonts}       
\usepackage[titletoc,title]{appendix}
\usepackage{bbm}            
\usepackage{bm}             
\usepackage{comment}
\usepackage{amsthm}         
\usepackage{makecell}
\usepackage{algorithm,algorithmic}
\usepackage{caption}
\makeatletter
\renewcommand{\ALG@name}{Protocol}
\makeatother
\usepackage[dvipsnames]{xcolor}
\definecolor{beamer@blendedblue}{rgb}{0.2,0.2,0.7}

\usepackage{booktabs} 

\usepackage[colorlinks=true,%
      linkcolor=beamer@blendedblue,%
      bookmarks=true,%
      breaklinks=true,%
      filecolor=beamer@blendedblue,%
      anchorcolor=yellow,%
      citecolor=beamer@blendedblue,%
      urlcolor=beamer@blendedblue,%
      pdfauthor={{TBD}},%
      pdfsubject={{TBD}},%
      CJKbookmarks=true]{hyperref}

\usepackage{float}    
\usepackage{verbatim}
\usepackage{tikz}
\usetikzlibrary{positioning, arrows.meta}

\usepackage{quantikz}

\newtheorem{definition}{Definition}
\newtheorem{proposition}[definition]{Proposition}

\newtheorem{lemma}[definition]{Lemma}
\newtheorem{theorem}[definition]{Theorem}

\mathchardef\ordinarycolon\mathcode`\:
\mathcode`\:=\string"8000
\def\vcentcolon{\mathrel{\mathop\ordinarycolon}}
\begingroup \catcode`\:=\active
  \lowercase{\endgroup
  \let :\vcentcolon
  }

\DeclareFontFamily{U}{mathx}{\hyphenchar\font45}
\DeclareFontShape{U}{mathx}{m}{n}{<-> mathx10}{}
\DeclareSymbolFont{mathx}{U}{mathx}{m}{n}
\DeclareMathAccent{\widebar}{0}{mathx}{"73}

\newcommand{\wt}[1]{\widetilde{#1}}
\newcommand{\wh}[1]{\widehat{#1}}

\DeclareMathOperator{\tr}{Tr}  

\newcommand{\id}{\operatorname{id}}

\makeatletter
\newsavebox{\@brx}
\newcommand{\llangle}[1][]{\savebox{\@brx}{\(\m@th{#1\langle}\)}%
  \mathopen{\copy\@brx\kern-0.5\wd\@brx\usebox{\@brx}}}
\newcommand{\rrangle}[1][]{\savebox{\@brx}{\(\m@th{#1\rangle}\)}%
  \mathclose{\copy\@brx\kern-0.5\wd\@brx\usebox{\@brx}}}
\makeatother

\newcommand*{\cD}{\mathcal{D}}

\newcommand*{\cF}{\mathcal{F}}

\newcommand*{\cM}{\mathcal{M}}
\newcommand*{\cN}{\mathcal{N}}
\newcommand*{\cO}{\mathcal{O}}

\newcommand*{\cS}{\mathcal{S}}
\newcommand*{\cT}{\mathcal{T}}

\newcommand{\bE}{\mathbb{E}}

\newcommand\eqnote[1]{\text{\color{blue}{\tiny\sffamily #1}}}
\usepackage{enumitem,amssymb,pifont}
\newlist{todolist}{itemize}{2}
\setlist[todolist]{label=$\square$}
\makeatletter
\newcommand{\appendixtitle}[1]{\gdef\@title{#1}}
\newcommand{\appendixauthor}[1]{\gdef\@author{#1}}
\newcommand{\appendixaffiliation}[1]{\gdef\@affiliation{#1}}
\newcommand{\appendixdate}[1]{\gdef\@date{#1}}
\makeatother

\makeatletter
\newcommand{\appendixmaketitle}{%
  \begin{center}%
    {\Large \@title \par}%
    \vspace{10pt}
    {\normalsize
      \lineskip .5em%
      \begin{tabular}[t]{c}%
        \@author
      \end{tabular}\par}%
    \vspace{5pt}
    {\itshape \@affiliation \par}%
    {\normalsize \@date}%
  \end{center}%
  \par
}
\makeatother

\begin{document}

\title{{Fidelity Estimation of Entangled Measurements with Local States}}

\author{Zanqiu Shen}
\email{shenzanqiu@baidu.com}
\affiliation{Institute for Quantum Computing, Baidu Research, Beijing 100193, China}

\author{Kun Wang}
\email{wangkun28@baidu.com}
\affiliation{Institute for Quantum Computing, Baidu Research, Beijing 100193, China}

\fontfamily{lmr}\selectfont

\begin{abstract}
We propose an efficient protocol to estimate the fidelity of 
an $n$-qubit entangled measurement device, 
requiring only qubit state preparations and classical data post-processing.
It works by measuring the eigenstates of Pauli operators, 
which are strategically selected according to their importance weights
and collectively contributed by all measurement operators.
We rigorously analyze the protocol's performance and
demonstrate that its sample complexity is
uniquely determined by the number of Pauli operators  
possessing non-zero expectation values with respect to the target measurement. 
Moreover, from a resource-theoretic perspective, 
we introduce the stabilizer R\'enyi entropy of quantum measurements 
as a precise metric to quantify the inherent difficulty of estimating measurement fidelity.
\end{abstract}

\maketitle
\tableofcontents

\setlength{\parskip}{\baselineskip}

\section{Introduction}
In recent years, substantial advancements have been made in constructing 
intermediate-scale quantum devices 
spanning diverse physical platforms~\cite{preskill2018quantum,o2009photonic, kielpinski2002architecture,jones2000geometric,casanova2016noise}. A pivotal phase in the fabrication of such devices involves assessing the performance of the measurement apparatus. This assessment can be executed through methodologies like measurement (detector) tomography \cite{fiuravsek2001maximum, lundeen2009tomography, hofmann2010complete, chen2019detector}, or through measurement fidelity estimation protocols, including those augmented by machine learning \cite{hentschel2010machine, magesan2013experimentally, lennon2019efficiently, nguyen2021deep}, or methods that establish bounds on average measurement fidelity \cite{magesan2013experimentally}. Nevertheless, complete tomography encounters challenges related to high sample complexity, rendering it impractical for large-scale systems. Conversely, measurement fidelity estimation protocols prove significantly more efficient, albeit their potential dependence on entangled state preparation. In scenarios necessitating entangled state preparation, the precision of the estimated fidelity can be notably affected by errors associated with state preparation, 
particularly in the noisy intermediate-scale quantum (NISQ) era~\cite{preskill2018quantum}.

In this work, we present an efficient protocol for estimating the fidelity of quantum measurement devices, exhibiting an exponential acceleration compared to quantum detector tomography and applicability to diverse quantum measurement schemes. First, we articulate the problem formulation. Consider a system comprising $n$ qubits, represented by a Hilbert space $\mathcal{H}$ with a dimension of $2^n$. Our protocol specifically targets the fidelity estimation of a general class of measurements—known as a projection-valued measure (PVM). By showcasing the efficacy of our protocol in certifying a wide range of measurement schemes, including Bell measurements and Elegant Joint Measurement (EJM) \cite{gisin2019entanglement, tavakoli2021bilocal}, we affirm its potential to significantly contribute to the advancement of experimental quantum information science.

Our devised protocol operates by conducting measurements on a randomly selected subset of Pauli observables, employing an importance sampling technique \cite{kliesch2021theory}. Specifically, we calculate the importance weights associated with all ideal PVMs relative to each Pauli operator. Through the accumulation of measurement statistics from this process, we can effectively estimate the measurement fidelity $\mathcal{F}(\mathcal{M}, \mathcal{N})$ between the ideal measurement $\mathcal{M}$ and its actual implementation $\mathcal{N}$. The preparation of product eigenstates for each Pauli observable necessitates repeated calls to the measurement device. 
The number of calls that the protocol makes to measurement device, 
which we shall call the sample complexity, 
depends on the specified measurement $\mathcal{M}$ of interest, 
the desired additive error $\varepsilon$, and the targeted failure probability denoted as $\delta$. 
We derive analytic bounds on the average sample complexity 
to elucidate the efficiency of our protocol. 
We compare our efficient protocol with two other protocols---the globally optimal estimation protocol that involves entangled state preparation and the ``direct'' measurement fidelity estimation protocol
adapted from the seminal direct fidelity estimation proposal by Flammia and Liu~\cite{flammia2011direct}
\footnote{The direct fidelity estimation protocol is proposed independently by Marcus P. da Silva, Olivier Landon-Cardinal and David Poulin~\cite{da2011practical}.}---and 
highlight our protocol's unique advantage in experimental feasibility and efficiency. 
Furthermore,
we conduct comprehensive numerical simulations to substantiate the effectiveness of our protocol.

\section{An efficient protocol}\label{sec:efficient protocol}

\subsection{Protocol description}

The fidelity between a known ideal $n$-qubit projective measurement $\cM$ 
and its actual implementation $\cN$ is defined as
\begin{align}
\label{def: measurement fidelity}
\mathcal{F}(\cM, \cN)
&:= \frac{1}{2^n} \sum_{k=1}^{2^n} \tr[\psi_k V_k],
\end{align}
where $\mathcal{M} \equiv \{ \psi_k=\proj{\psi_k}\}_k$ represents the ideal PVM 
satisfying $\sum_{k=1}^{2^n} \psi_k = \id$, 
where $\id$ is the identity operator,
while $\mathcal{N} \equiv \left\{ V_k \right\}_k$ 
denotes the actual positive operator-valued measure (POVM) 
satisfying $0\leq V_k\leq\id$ and $\sum_{k=1}^{2^n} V_k = \id$. 
Note that the specification of $\cM$ is known to the experimenters 
but the specification of $\cN$ is completely unknown.
The task is to estimate the fidelity $\mathcal{F}(\cM, \cN)$ given accesses to $\cN$.
In the following, we construct an efficient protocol 
to accomplish this task using local state preparations
and rigorously analyze the protocol's performance.

Expanding the measurement operators $\psi_k$ in terms of the set of 
normalized Pauli operators $\{P_k\}_k, k=1,2,...,4^n$, 
which constitutes an orthonormal basis of the $n$-qubit operators, we obtain
\begin{align}
    \cF(\cM, \cN)
=  \frac{1}{2^n} \sum_{k=1}^{2^n} 
    \tr\left[\left(\sum_{l=1}^{4^n}\tr[\psi_k P_l]P_l\right)V_k\right]
= \frac{1}{2^n} \sum_{k=1}^{2^n} \sum_{l=1}^{4^n} \tr[\psi_k P_l] \tr[V_k P_l].
    \label{double sum}
\end{align}
In order to construct an estimator for $\cF(\cM, \cN)$, 
we define a probability mass function $\{q_l\}_{l=1}^{4^n}$ with respect to the Pauli operators as
\begin{align}
    q_l
    &:= \frac{\sum_{k=1}^{2^n} \tr[\psi_k P_l]^2}{2^n}.
    \label{def: importance sampling}
\end{align}
Thanks to the property of PVM, we can show that $\sum_{l=1}^{4^n} q_l = 1$ 
and thus $\{q_l\}_l$ is a valid probability mass function. 
Now, we employ the importance sampling technique \cite{kliesch2021theory} 
to estimate $\cF$ via Eq.~\eqref{double sum}. 
To this end, we rewrite $\cF$ as 
\begin{align}
\cF 
= \sum_{l=1}^{4^n}\left(\sum_{k'=1}^{2^n}\frac{\tr[\psi_{k’}P_l]^2}{2^n}\right)
  \times \left(\sum_{k=1}^{2^n}\tr[V_kP_l] 
  \frac{ \tr[\psi_kP_l]}{\sum_{k'=1}^{2^n}\tr[\psi_{k’}P_l]^2}\right)
\equiv \sum_{l=1}^{4^n} q_l X_l,\label{eq:new quantity}
\end{align}
where the random variable $X_l$ is defined as
\begin{align}\label{eq:random-variable-X-l}
    X_l := \frac{\sum_{k=1}^{2^n}\tr[V_k P_l] \tr[\psi_k P_l]}{\sum_{k'=1}^{2^n}\tr[\psi_{k’}P_l]^2}.
\end{align}
We can show that $\bE[X_l] = \cF$, 
where $\bE[\cdot]$ denotes the expectation value, 
signifying that $X_l$ serves as an unbiased estimator of $\cF$.
From Eq.~\eqref{eq:new quantity}, 
we can construct an estimator of $\cF$ using the sample mean in the following manner.
Initially, $m$ indices $l_i$, where $i = 1, 2, ..., m$ and $l_i \in \{1, 2, ..., 4^n\}$, are sampled according to the importance sampling 
distribution $\{q_l\}_l$ defined in Eq.~\eqref{def: importance sampling}, 
completely determined by the known ideal PVM $\cM$, where 
\begin{align}
    m = \left\lceil 1/(\varepsilon^2 \delta) \right\rceil
    \label{samples: m}
\end{align}
and $\lceil \cdot \rceil$ denotes rounding up. 
Given the estimated values $\{X_{l_i}\}_{i=1}^m$, 
whose estimation procedure will be presented below, 
we can construct an estimator of $\cF$ using the following formula:
\begin{align}\label{ey}
    Y = \frac{1}{m}\sum_{i=1}^m X_{l_i}.
\end{align}
By Chebyshev's inequality (see Appendix~\ref{appx: the efficient protocol} for detailed arguments), 
we assert that $Y$ is an $\varepsilon$-accurate estimator of $\cF$ satisfying
\begin{align}
    {\rm Pr}\left(\vert Y - \cF \vert \geq \varepsilon \right) \leq \delta,
    \label{estimation step 1}
\end{align}
where $\varepsilon$ and $\delta$ are predefined constant additive error and failure probability, respectively.

It is noteworthy that for each randomly chosen $l$, 
the quantity $\tr[V_k P_l]$ need to be estimated from the measurement statistics in order to determine $X_l$. 
Since $P_l$ is not a quantum state, it cannot be used as an input to the measurement device directly.
Thanks to the linearity of the trace function, we can consider the spectral decomposition
\begin{align}\label{eq:spectral-decomposition}
    P_{l} 
= \sum_{a=1}^{2^n}\lambda_{a}^{(l)} {\vert \phi_{a}^{(l)} \rangle \langle \phi_{a}^{(l)} \vert} 
\equiv \sum_{a=1}^{2^n}\lambda_{a}^{(l)} \phi_{a}^{(l)},
\end{align}
where $\{\lambda_{a}^{(l)}\}_a$ and $\{\phi_{a}^{(l)}\equiv\vert \phi_{a}^{(l)} \rangle \langle \phi_{a}^{(l)} \vert\}_a$ 
are the eigenvalues and eigenstates of $P_l$, respectively.
Based on the spectral decomposition, 
we can rewrite $X_l$ in Eq.~\eqref{eq:random-variable-X-l} as follows:
\begin{align}\label{eq:X-l}
    X_l = \sum_{a=1}^{2^n}\lambda_{a}^{(l)}
   \left(\sum_{k=1}^{2^n}\frac{\tr[V_k \phi_{a}^{(l)}]\tr[\psi_kP_l]}{\sum_{k'=1}^{2^n}\tr[\psi_{k'}P_l]^2}\right).
\end{align}
Since each eigenstate $\phi_{a}^{(l)}$ is a product state, 
it can be prepared with high fidelity in NISQ quantum devices.
Correspondingly, $\tr[V_k \phi_{a}^{(l)}]$ can be estimated 
by first measuring $\phi_{a}^{(l)}$ many times with the measurement device $\cN$ and 
then post-processing the measurement statistics.
Specifically, we consider the following subroutine:
\begin{enumerate}
    \item \textit{Sample an index.} Choose an integer $a_j\in\{1,\cdots,2^n\}$ uniformly at random;
    \item \textit{Measure the eigenstate.}
          Use the measurement device $\cN$ to measure 
          the eigenstate $\phi^{(l_i)}_{a_j}$ (corresponding to the Pauli operator $P_{l_i}$) 
          and record the measurement outcome $o_j$.
\end{enumerate}
We define a new random variable as follows:
\begin{align}
\label{wtXl}
    \wh{X}_{a_j}^{(l_i)} := 
\frac{2^n}{\sum_{k'=1}^{2^n}\tr[\psi_{k'}P_{l_i}]^2}\times\lambda^{(l_i)}_{a_j}\times  \tr[\psi_{o_j}P_{l_i}].
\end{align}
It is established that $\bE[\wh{X}_{a_j}^{(l_i)}] = X_{l_i}$ as demonstrated in Appendix~\ref{appx: the efficient protocol}. 
By executing this procedure a total of $n_{l_i}$ times, 
we can estimate $X_{l_i}$ to desired accuracy via the estimator
\begin{align}
    \wt{X}_{l_i} := \frac{1}{n_{l_i}} \sum_{j=1}^{n_{l_i}} \wh{X}_{a_j}^{(l_i)}.
\end{align}

Finally, we introduce the following estimator:
\begin{align}
\label{wtY-ptl-1}
    \wt{Y} 
= \frac{1}{m} \sum_{i=1}^m \wt{X}_{l_i}
=   \sum_{i=1}^m\sum_{j=1}^{n_{l_i}}
    \frac{1}{mn_{l_i}}\times\frac{2^n}{\sum_{k'=1}^{2^n}\tr[\psi_{k’}P_{l_i}]^2}
    \times\lambda^{(l_i)}_{a_j} \times\tr[\psi_{o_j}P_{l_i}].
\end{align}
To meet the confidence interval requirement that
\begin{align}
    {\rm Pr}\left(\vert \wt{Y} - Y \vert \geq \varepsilon \right) \leq \delta,
    \label{estimation step 2}
\end{align}
we choose
\begin{align}\label{nl}
n_{l_i} &=  \left\lceil \frac{1}{\left(\sum_{k'=1}^{2^n}\tr[\psi_{k’}P_{l_i}]^2\right)^2}
        \times\frac{2}{m\varepsilon^2}\ln\frac{2}{\delta} \right\rceil,
\end{align}
where $\ln$ denotes the natural logarithm.
Using the confidence intervals (\ref{estimation step 1}) and (\ref{estimation step 2}) 
and the union bound, we establish a confidence interval of the final estimation $\wt{Y}$ as
\begin{align}\label{est: eps-close est for efficient}
    {\rm Pr}\left\{\vert \wt{Y} - \cF \vert \geq 2 \varepsilon \right\} \leq 2 \delta.
\end{align}

We summarize the random variables and estimation chain involved in our 
estimation protocol in Fig.~\ref{fig: estimation process}
and the complete efficient protocol in \textbf{Protocol~\ref{ptl: efficient protocol}} for reference. 
For extensive details and rigorous proofs regarding the aforementioned protocol, 
consult Appendix~\ref{appx: the efficient protocol}.

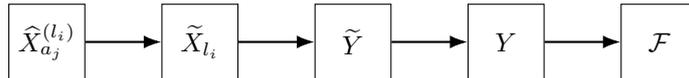
\begin{figure}[!hbtp]
  \centering
  \begin{tikzpicture}[
    node distance=1cm,
    arrow/.style={-Latex, thick},
    box/.style={draw, rectangle, minimum width=1cm, minimum height=1cm, align=center}
  ]

    \node[box] (Xl) {$\wh{X}_{a_j}^{(l_i)}$};
    \node[box, right=of Xl] (wX) {$\wt{X}_{l_i}$};
    \node[box, right=of wX] (wY) {$\wt{Y}$};
    \node[box, right=of wY] (Y) {${Y}$};
    \node[box, right=of Y] (F) {$\mathcal{F}$};

    \draw[arrow] (Xl) -- (wX);
    \draw[arrow] (wX) -- (wY);
    \draw[arrow] (wY) -- (Y);
    \draw[arrow] (Y) -- (F);

  \end{tikzpicture}
  \caption{The random variables and estimation chain involved in our measurement fidelity estimation protocol.}
  \label{fig: estimation process}
\end{figure}

\begin{algorithm}[H]
\caption{An efficient protocol}
\label{ptl: efficient protocol}
\begin{algorithmic}[1]
\REQUIRE $\cM\equiv\{\psi_k\}_k$: the PVM that the $n$-qubit measurement device is designed to implement,\\
\hskip1.9em $\cN\equiv\{V_k\}_k$: the POVM that the $n$-qubit measurement device actually implements, \\
\hskip1.9em $\varepsilon$: the constant additive error, determined by the experimenter, and \\
\hskip1.9em $\delta$: the failure probability, determined by the experimenter.
\ENSURE Estimation of the measurement fidelity $\cF(\cM,\cN)$.
\STATE Calculate the required number of Pauli operators $m = \lceil {1}/{\varepsilon^2 \delta} \rceil$,
\STATE Randomly sample $m$ Pauli operators $P_{l_i}$, $i = 1, 2,..., m$
        with respect to the probability distribution $q_l$ defined in Eq.~\eqref{def: importance sampling}.
\FOR{$i = 1, \cdots, m$}
\STATE Calculate the required number of measurements $n_{l_i}$ using Eq.~\eqref{nl},
\FOR{$j = 1, \cdots, n_{l_i}$}
\STATE Choose some $a_j\in\{1,\cdots,2^n\}$ uniformly at random;
\STATE Use the quantum measurement device to measure the eigenstate $\phi^{(l_i)}_{a_j}$ (corresponding to $P_{l_i}$) 
and record the experimental outcome $o_j$.
\ENDFOR
\STATE Calculate and record $\wt{X}_{l_i}$.
\ENDFOR
\STATE Use the dataset $\{\wt{X}_{l_i}\}_{i=1}^{m}$ to calculate $\wt{Y}$ using Eq.~\eqref{wtY-ptl-1}.
\STATE Output the measurement fidelity estimation $\wt{Y}$
        satisfying ${\rm Pr}(\vert \wt{Y} - \cF \vert \leq 2 \varepsilon) \geq 1 - 2 \delta$.
\end{algorithmic}
\end{algorithm}

It is crucial to note that $\wh{X}_{a_j}^{(l_i)}$ in Eq.~\eqref{wtXl} is always nonzero 
since all possible measurement outcomes are taken into account. 
This implies that we fully leverage the measurement results for the estimation without information loss, 
thus achieving a reduction of the sample complexity by a factor of $2^n$. 
In the following Section~\ref{sec:direct protocol},
we will present an alternative protocol,
which is ``direct'' measurement fidelity estimation protocol
adapted from the seminal direct fidelity estimation 
proposal by Flammia and Liu~\cite{flammia2011direct},
showing that it suffers from information loss and 
is thus less efficient compared to \textbf{Protocol~\ref{ptl: efficient protocol}} presented here.

\subsection{Performance analysis}

In order to establish a theoretical guarantee for the above estimation protocol, 
it is imperative to ensure that the accuracy of the estimation in two steps. 
First, we analyze the accuracy of the estimator $Y$ that relies on a finite number of Pauli observables. 
The precision of $Y$ as an estimator of $\cF$ is governed by the variance of $X_l$ 
and can be enhanced by increasing the number of sampled Pauli operators $m$.
Subsequently, we delve into the analysis of the accuracy of the 
estimator $\wt{Y}$ of $Y$, which incorporates the statistical error incurred by finite measurement outcomes.

We adopt the \emph{total number of calls} of the measurement device $\cN$, defined as
\begin{align}\label{eq:sample-complexity}
    L := \sum_{i=1}^m n_{l_i},
\end{align}
where $n_{l_i}$ is defined in Eq.~\eqref{nl}, as the metric to quantify the protocol's performance. 
This quantity is also termed the sample complexity of the protocol.
Combining Eqs.~\eqref{samples: m} and~\eqref{nl}, 
we establish a lower bound on the expected value of $L$,
showing that the performance of our protocol is
uniquely determined by the number of Pauli operators  
possessing non-zero expectation values with respect to the ideal PVM $\cM$.
The detailed proof is provided in Appendix \ref{appx: the efficient protocol}.
\begin{theorem}
\label{thm: sc of efficient protocol}
Let $\cM\equiv\{\psi_k\}_k$ be an ideal PVM.
To achieve a $2\varepsilon$-close estimator $\wt{Y}$ such that 
${\rm Pr}(\vert \wt{Y} - \cF \vert \leq 2 \varepsilon) \geq 1 - 2 \delta$, 
the number of measurement device calls $L$ in \textbf{Protocol \ref{ptl: efficient protocol}}, 
defined in Eq.~\eqref{eq:sample-complexity}, must satisfy
  \begin{align}
    \mathbb{E}[L] \geq {\frac{\vert \cS \vert^2}{4^n}}\frac{2}{\varepsilon^2}\ln\frac{2}{\delta},
  \end{align}
  where $\mathcal{S}$ is the set of Pauli operators that have non-zero expectation values with respect to $\cM$:
  \begin{align}\label{set s}
    \mathcal{S} := \left\{l \;\middle\vert\;
                \sum_{k=1}^{2^n}\tr[\psi_{k}P_l]^2 \neq 0, \forall l \right\},
    \end{align}
  and $\vert \cS \vert$ denotes the cardinality of the set $\cS$.
\end{theorem}

Theorem~\ref{thm: sc of efficient protocol} offers valuable insights into the 
resource requirements for attaining precise estimations, 
taking into consideration the dimensionality of the system. 
It is possible the cardinality $\vert \mathcal{S} \vert$ is significantly less than $4^n$, 
thereby enabling superior performance.

\subsection{Bounds by stabilizer R\'enyi entropy}
\label{sec:bound-by-stabilizer-for-efficient-ptl}

Motivated by the 
insights presented in~\cite{leone2023nonstabilizerness}, 
we introduce the stabilizer R\'enyi entropy of quantum measurements 
as a quantitative measure for evaluating 
the resource requirements associated with the performance assessment of 
a quantum measurement device.
We define an observable $O$ linked to the ideal PVMs $\cM\equiv\{\psi_k\}_k$ as follows:
\begin{align}
\label{def: o}
    O
    &:= \sum_{l=1}^{4^n} \sqrt{\frac{\sum_{k=1}^{2^n} \tr[\psi_k P_l]^2}{2^n}} P_l.
\end{align}
Leveraging $O$, we define the \emph{$\alpha$-stabilizer R\'{e}nyi entropy} of $\cM$ as
\begin{align}
\label{def: alpha sre}
    M_{\alpha}(\cM)
    &:= R_{\alpha}(O) - \ln \tr[O^2] - \ln 2^n,
\end{align}
where 
\begin{align}
    R_{\alpha}(O) := \frac{1}{1- \alpha} \ln \left[\sum_{P\in\cS} \tr[OP]^{2 \alpha}\right]
\end{align}
and the set $\cS$ is defined in Eq.~\eqref{set s}.
We note that the $\alpha$-stabilizer R\'{e}nyi entropy of quantum states 
and quantum unitaries has been defined and investigated in~\cite{leone2022stabilizer}. 
We extend their definitions to the quantum measurements.

Interestingly, we endow the $\alpha$-stabilizer R\'enyi entropy measure $M_{\alpha}(\cM)$ 
with an operational interpretation within the fidelity estimation task
by showing that it induces lower and upper bounds on
our proposed estimation protocol's sample complexity.
This result establishes the fundamental limits and delivered meaningful insights to 
quantum measurement fidelity estimation schemes using the importance sampling technique.
The detailed proof can be found in Appendix~\ref{appx: magic bounds for efficient protocol}.

\begin{theorem}
\label{thm: magic bounds for new method}
    Let $\cM\equiv\{\psi_k\}_k$ be an ideal PVM. 
    To achieve a $2\varepsilon$-close estimator $\wt{Y}$ such that 
    ${\rm Pr}(\vert \wt{Y} - \cF \vert \leq 2 \varepsilon) \geq 1 - 2 \delta$, 
    the number of measurement device calls $L$ in \textbf{Protocol \ref{ptl: efficient protocol}}, 
    defined in Eq.~\eqref{eq:sample-complexity}, must satisfy
    \begin{align}
        \frac{1}{2 \varepsilon^2} \ln\frac{1}{\delta} \exp\left[2 M_2(\cM) \right] \leq L \leq \frac{2^4}{\varepsilon^6} \ln\frac{1}{\delta} \exp[2M_0(\cM)].
    \end{align}
\end{theorem}

\section{Many more protocols}

Broadening our exploration of measurement fidelity estimation, we unveil two additional protocols: 
a globally optimal protocol that harnesses the power of entangled state preparation, 
and a protocol directly adapted from the pioneering direct fidelity estimation proposal 
by Flammia and Liu~\cite{flammia2011direct} (this is why we call it a ``direct'' protocol) . 
These protocols, meticulously analyzed for their complexity, 
paint a diverse landscape of estimation techniques.
Delving into the genesis of Protocol~\ref{ptl: efficient protocol}, 
we trace its roots to these global and direct protocols, revealing the motivations 
and insights that led to its development. Culminating this journey, 
we present a comprehensive comparison of the different protocols, 
offering a comprehensive guide for experimenters navigating the realm of measurement fidelity estimation.

\subsection{Global protocol}\label{sec:global protocol}

Here we introduce the global protocol, where
we can prepare entangled states as inputs and achieve the globally optimal estimation efficiency.
The measurement fidelity $\cF$,  
originally defined in Eq.~\eqref{def: measurement fidelity}, can be rewritten as
\begin{align}\label{eq:global-fidelity}
\cF &:= \sum_{k=1}^{2^n}\frac{1}{2^n} \times \tr[\psi_k V_k]  \equiv \sum_{k=1}^{2^n} q_k X_k,
\end{align}
where $q_k := 1/2^n$ and $X_k:= \tr[\psi_k V_k]$. 
Subsequently, we randomly select $k_i$, where $k_i \in \{1, 2, ..., 2^n\}$ according to $\{q_k\}_k$ 
and perform a one-shot measurement on $\psi_{k_i}$. 
Define $\wt{X}_{k_i}:=\delta_{k_i o_i}$ as the estimation of $X_{k_i}$, 
where $o_i$ is the measurement outcome. 
Essentially, this is a simple application of the importance sampling technique 
with respect to the uniform distribution.
We conclude that 
\begin{align}
\label{wtY}
\wt{Y} := \frac{1}{L} \sum_{i=1}^L \wt{X}_{k_i}
\end{align}
is an vaild estimator of $\cF$. 

\textbf{Protocol \ref{ptl: naive protocol}} outlines the overall estimation process and 
we conclude the sample complexity of this global protocol in the following theorem,
whose proof can be found in Appendix~\ref{appx: naive protocol}.
\begin{theorem}
  \label{thm: optimal lower bound for measurement fidelity}
  Let $\cM\equiv\{\psi_k\}_k$ be an ideal PVM.
  To achieve a $2\varepsilon$-close estimator $\wt{Y}$ such that 
  ${\rm Pr}(\vert \wt{Y} - \cF \vert \leq 2 \varepsilon) \geq 1 - 2 \delta$, 
  the total number of measurement device calls
  in the globally optimal \textbf{Protocol \ref{ptl: naive protocol}} must satisfy
  \begin{align}
    L &\geq \frac{1}{8 \varepsilon^2} \ln\frac{1}{\delta}.
  \end{align}
\end{theorem}

Note that the sample complexity of the global protocol is 
independent of the dimension of the system, which is appealing in verification efficiency.
However, it necessitates entangled state preparation, 
which introduces a substantial preparation error and is experimentally challenging. 
In such cases, the reliability of the estimation is significantly compromised. 
Therefore, for quantum fidelity estimation tasks, we prefer product state preparations.

\begin{algorithm}[H]
\caption{The globally optimal protocol}
\label{ptl: naive protocol}
\begin{algorithmic}[1]
\REQUIRE $\cM\equiv\{\psi_k\}_k$: the PVM that the $n$-qubit measurement device is designed to implement,\\
\hskip1.9em $\cN\equiv\{V_k\}_k$: the POVM that the $n$-qubit measurement device actually implements, \\
\hskip1.9em $\varepsilon$: the constant additive error, determined by the experimenter, and \\
\hskip1.9em $\delta$: the failure probability, determined by the experimenter.
\ENSURE Estimation of the measurement fidelity $\cF(\cM,\cN)$.
\STATE Calculate the required number of samples:
\begin{align}
    L = \left\lceil \frac{1}{8 \varepsilon^2} \ln\frac{1}{\delta} \right\rceil.
\end{align}
\FOR{$i = 1, \cdots, L$}
\STATE Uniformly sample an ideal PVM operator $\psi_{k_i}$, where $k_i \in \{1, 2, ..., 2^n\}$;
\STATE Measure $\psi_{k_i}$ with the measurement device and record outcome $o_i$.
        Let $\wt{X}_{k_i} = \delta_{k_i o_i}$.
\ENDFOR
\STATE Using the dataset $\{ \wt{X}_{k_i}\}_{i=1}^L$, compute the estimator $\wt{Y} $as Eq.~\eqref{wtY}.
\STATE Output the measurement fidelity estimation $\wt{Y}$
        satisfying ${\rm Pr}(\vert \wt{Y} - \cF \vert \leq 2 \varepsilon) \geq 1 - 2 \delta$.
\end{algorithmic}
\end{algorithm}

\subsection{Direct protocol}\label{sec:direct protocol}

In this section, we adopt the concept of Direct Fidelity Estimation (DFE) for entangled states in \cite{flammia2011direct,da2011practical} and tailor their core idea to measurement fidelity estimation
in a ``direct'' manner. This direct protocol exclusively requires the preparation of product eigenstates 
associated with Pauli observables, almost in the same way as the protocol designed in Section~\ref{sec:efficient protocol}.
However, there a critical difference in designing the importance sampling distributions,
which we will rigorously prove, leads to a significant gap in estimation efficiency between these two protocols.

As always, we rewrite the measurement fidelity $\cF$ as
\begin{align}
\cF
&= \sum_{k=1}^{2^n} \sum_{l=1}^{4^n} \frac{\tr[\psi_k P_{l}]^2}{2^n} \times \frac{\tr[V_k P_{l}]}{\tr[\psi_k P_{l}]}
\equiv \sum_{k=1}^{2^n} \sum_{l=1}^{4^n} q_{kl} X_{kl},
\end{align}
where $q_{kl}:= {\tr[\psi_k P_{l}]^2}/{2^n}$ forms a joint probability mass function and
\begin{align}\label{eq:X-kl}
X_{kl}:= \frac{\tr[V_k P_l]}{\tr[\psi_k P_{l}]}   
\end{align}
denotes a random variable. 
Leveraging the sampling distribution $\{q_{kl}\}_{kl}$ 
and drawing inspiration from the methodology employed in the construction of 
the efficient \textbf{Protocol \ref{thm: sc of efficient protocol}}, 
we provide a succinct summary of the direct protocol in \textbf{Protocol \ref{ptl: original protocol}}
and a more comprehensive elucidation can be found in Appendix \ref{appx: the original protocol}.  

We establish a upper bound on the expected value of the number of calls $L$ 
in \textbf{Protocol \ref{ptl: original protocol}}
and the proof can be found in Appendix~\ref{appx:sc of original protocol}.
\begin{theorem}
\label{thm: sc of original protocol}
Let $\cM\equiv\{\psi_k\}_k$ be an ideal PVM. 
To achieve a $2\varepsilon$-close estimator $\wt{Y}$ such that 
${\rm Pr}(\vert \wt{Y} - \cF \vert \leq 2 \varepsilon) \geq 1 - 2 \delta$, 
the expected number of measurement device calls 
in \textbf{Protocol \ref{ptl: original protocol}} satisfies
\begin{align}
\mathbb{E}[L] \leq 1+\frac{1}{\varepsilon^2 \delta} + 2 \vert \cT \vert \frac{1}{\varepsilon^2} \ln\frac{2}{\delta},
\end{align}
where
\begin{align}
\label{set T}
\cT &:= \left\{ (k, l) \;\middle\vert\; \tr[\psi_k P_{l}] \neq 0, \forall k, l \right\},
\end{align}
and $\vert \cT \vert$ denotes the cardinality of set $\cT$.
\end{theorem} 

\begin{algorithm}[H]
\caption{A direct protocol}
\label{ptl: original protocol}
\begin{algorithmic}[1]
\REQUIRE $\cM\equiv\{\psi_k\}_k$: the PVM that the $n$-qubit measurement device is designed to implement,\\
\hskip1.9em $\cN\equiv\{V_k\}_k$: the POVM that the $n$-qubit measurement device actually implements, \\
\hskip1.9em $\varepsilon$: the constant additive error, determined by the experimenter, and \\
\hskip1.9em $\delta$: the failure probability, determined by the experimenter.
\ENSURE Estimation of the measurement fidelity $\cF(\cM,\cN)$.
\STATE Calculate the required number of Pauli operators $m = \lceil {1}/{\varepsilon^2 \delta} \rceil$,
\STATE Jointly sample $m$ Pauli operators $\{P_{l_i}\}_{i=1}^m$ and $m$ ideal PVM elements $\{\psi_{k_i}\}_{i=1}^m$ with the joint probability mass function $q_{kl}:= {\tr[\psi_k P_{l}]^2}/{2^n}$.
\FOR{$i = 1, \cdots, m$}
\STATE Calculate the required number of measurements
    $n_i=\lceil 2^{n+1}/(m \varepsilon^2\tr[\psi_{k_i} P_{l_i}]^2)\ln(2/\delta)\rceil$.
\FOR{$j = 1, \cdots, n_{l_i}$}
\STATE Choose some $a_j\in\{1,\cdots,2^n\}$ uniformly at random,
\STATE Use the quantum measurement device to measure the eigenstate $\phi^{(l_i)}_{a_j}$ (corresponding to $P_{l_i}$) and record the experimental outcome $o_j$.
\ENDFOR
\STATE Calculate and record $\wt{X}_{k_i l_i}$.
\ENDFOR
\STATE Use the dataset $\{\wt{X}_{k_i l_i}\}_{i=1}^m$ and calculate $\wt{Y}$.
\STATE Output the measurement fidelity estimation $\wt{Y}$
        satisfying ${\rm Pr}(\vert \wt{Y} - \cF \vert \leq 2 \varepsilon) \geq 1 - 2 \delta$.
\end{algorithmic}
\end{algorithm}

We observe that the maximum cardinality of $\vert \cT \vert$ is $8^n$, indicating that the worst-case sample complexity scales as $\cO(8^n)$. While this represents an improvement of a factor of $2^n$ compared to 
standard measurement tomography, a substantial performance gap still exists when compared to the global \textbf{Protocol \ref{ptl: naive protocol}}. 
It is evident that there is a trade-off between entangled state preparation and sample complexity. 
In the efficient \textbf{Protocol \ref{ptl: efficient protocol}}, 
we further reduce the sample complexity while relying solely on product state preparation.

Similar to Section \ref{sec:bound-by-stabilizer-for-efficient-ptl}, 
we can bound the sample complexity of the direct \textbf{Protocol \ref{ptl: original protocol}} 
by the stabilizer R\'enyi entropy $M_\alpha(\cM)$ defined in Eq.~\eqref{def: alpha sre} 
in the following theorem. 
The detailed proof can be found in Appendix~\ref{appx: magic bounds for original protocol}.
It can be seen that the sample complexity is bounded from below by $\max_k \exp[M_2(\psi_k)] $ 
and from below by $\max_k \exp[M_0(\psi_k)]$. 

\begin{theorem}
\label{thm: magic bounds for original protocol}
Let $\cM\equiv\{\psi_k\}_k$ be an ideal PVM. 
To achieve a $2\varepsilon$-close estimator $\wt{Y}$ such that 
${\rm Pr}(\vert \wt{Y} - \cF \vert \leq 2 \varepsilon) \geq 1 - 2 \delta$, 
the number of measurement device calls $L$ in \textbf{Protocol \ref{ptl: original protocol}} must satisfy
\begin{align}
  \frac{4^n}{2 \varepsilon^2} \ln\frac{1}{\delta} \max_k \exp[M_2(\psi_k)] 
  &\leq L \leq \frac{4 \cdot 4^n}{\varepsilon^4} \ln\frac{1}{\delta} \max_k \exp[M_0(\psi_k)].
\end{align}
\end{theorem}

\subsection{Comparisons of the protocols} 

In this work, we have introduced three different measurement fidelity estimation protocols:
\begin{itemize}
\item \textbf{Protocol~\ref{ptl: naive protocol}} in Section~\ref{sec:efficient protocol} 
        achieves the optimal estimation efficiency, 
        yet it necessitates entangled state preparation and is thus experimentally challenging.
\item \textbf{Protocol~\ref{ptl: original protocol}} in Section~\ref{sec:direct protocol} 
        is a ``direct'' measurement fidelity estimation protocol
        adapted from the seminal direct fidelity estimation proposal by Flammia and Liu~\cite{flammia2011direct}.
        It requires only local state preparations but its sample complexity is as high as $\cO(8^n)$,
        where $n$ is the number of qubits of the target measurement.
\item \textbf{Protocol~\ref{ptl: efficient protocol}} in Section~\ref{sec:efficient protocol} 
        shares the benefit of \textbf{Protocol~\ref{ptl: original protocol}}
        in requiring only local state preparations.
        Remarkably, it achieves a significant breakthrough in efficiency, 
        boosting it to $\Omega(\vert\cS\vert^2/4^n)$ thanks to a novel importance sampling 
        distribution that avoids information loss. 
\end{itemize}
Table~\ref{tbl:comparison} paints a clear picture: 
all our estimation protocols outperform the standard Quantum Detector Tomography (QDT) by 
a significant factor in terms of required samples as expected. 
\textbf{Protocols~\ref{ptl: efficient protocol}} and~\textbf{\ref{ptl: original protocol}} truly shine, 
not only demonstrating the better estimation efficiency but also remaining experimentally friendly 
as they require only local state preparations. 
This combination of superior performance and practical feasibility 
makes them compelling choices for real-world measurement fidelity estimation tasks.

\begin{table}[!hbtp]
\centering
\setlength{\tabcolsep}{8pt}
\setlength\heavyrulewidth{0.3ex}
\renewcommand{\arraystretch}{1.5}
\begin{tabular}{@{}ccc@{}}
    \toprule
        \textbf{Protocol} & \textbf{Sample complexity} & \textbf{Remark} \\\hline
        \textbf{Protocol~\ref{ptl: efficient protocol} in Section~\ref{sec:efficient protocol}} 
            & {$\Omega(\vert\cS\vert^2/4^n)$}  & local states \\
        \textbf{Protocol~\ref{ptl: naive protocol} in Section~\ref{sec:global protocol}} 
            & $\cO(1)$  & entangled states \\
        \textbf{Protocol~\ref{ptl: original protocol} in Section~\ref{sec:direct protocol}} 
            & $\cO(8^n)$ & local states \\
        \textbf{Quantum detector tomography~\cite{surawy2022projected}} 
            & $\cO(16^n)$ & local states \\
    \bottomrule
\end{tabular}
\caption{Comparison of different protocols for the measurement fidelity estimation task.
Here, $n$ is the number of qubits of the target measurement $\cM$ 
and $\mathcal{S}$ is the set of $n$-qubit Pauli operators that have non-zero expectation values 
with respect to $\cM$; cf. Eq.~\eqref{set s} for definition.}
\label{tbl:comparison}
\end{table}

To further explore the intrinsic differences between \textbf{Protocol \ref{ptl: efficient protocol}} and 
\textbf{Protocol \ref{ptl: original protocol}}, apart from the apparent difference in sample complexity, 
we compute and compare their Root Mean Square Errors (RMSE) and 
the conclusion is drawn in the following proposition.
This result will be validated numerically in the following section 
and the detailed proof can be found in Appendix~\ref{appx: prop. 8}.

\begin{proposition}
\label{coro: var comp}
Under the assumption that $X_l$~\eqref{eq:X-l} and $X_{kl}$~\eqref{eq:X-kl} can be perfect estimated,
the RMSE of \textbf{Protocol~\ref{ptl: efficient protocol}} takes the form of 
$\mathcal{O}(c_1/\sqrt{m})$ and the RMSE of \textbf{Protocol~\ref{ptl: original protocol}} 
takes the form of $\mathcal{O}(c_3/\sqrt{m})$, where $m$ is the number of sampled Pauli operators. 
It holds that $c_1 \leq c_3$, indicating that \textbf{Protocol \ref{ptl: efficient protocol}} 
is conclusively more efficient than or equally efficient to \textbf{Protocol \ref{ptl: original protocol}}.
\end{proposition}

\section{Numerics}

To comprehensively validate the precision of our protocols in tackling real-world challenges, 
we embark on a rigorous series of numerical simulations. 
With a focus on estimating the fidelity of Bell measurements as shown in Fig.~\ref{fig:bell-measurement}, 
a pillar of quantum teleportation and dense coding, 
these simulations unveil intriguing insights into the performance of our protocols 
and their potential to revolutionize quantum information tasks.

\begin{figure}[!htbp]
    \centering
    \begin{quantikz}
        \lstick[wires=2]{$\ket{\phi}$} & \ctrl{1} & \gate{H}   & \meter{} \\
        \lstick{}                      & \targ{}  & \qw        & \meter{} \\
    \end{quantikz}
    \caption{An illustrative quantum circuit for implementing the Bell measurement.
            The two-qubit controlled gate is the CX gate and 
            the two single-qubit measurements are the computational basis measurements.
            Here $\varepsilon=\delta=0.05$.}
    \label{fig:bell-measurement}
\end{figure}
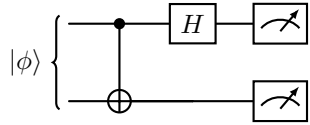

Specific accuracy parameters, namely $\varepsilon = \delta = 0.05$, are set as standards. 
Notably, depolarizing errors were intentionally introduced on the two-qubit CX gate. 
These errors are precisely defined by the depolarizing operator $\cD(\rho) = (1 - p)\rho + p \tr[\rho] \cdot {\id}/{2^n}$, where $p$ represents the depolarizing error rate and $\rho$ signifies the input state.
The results, depicted in Fig.~\ref{fig1}, reveal a remarkable agreement between the estimated 
fidelity and the corresponding theoretical fidelity across diverse depolarizing error rates. 
Furthermore, the estimated fidelity consistently conforms to the predefined upper and lower bounds, 
as articulated in Eq.~\eqref{est: eps-close est for efficient}.

\begin{figure}[!hbtp]
	\centering
	\includegraphics[width=0.8\textwidth]{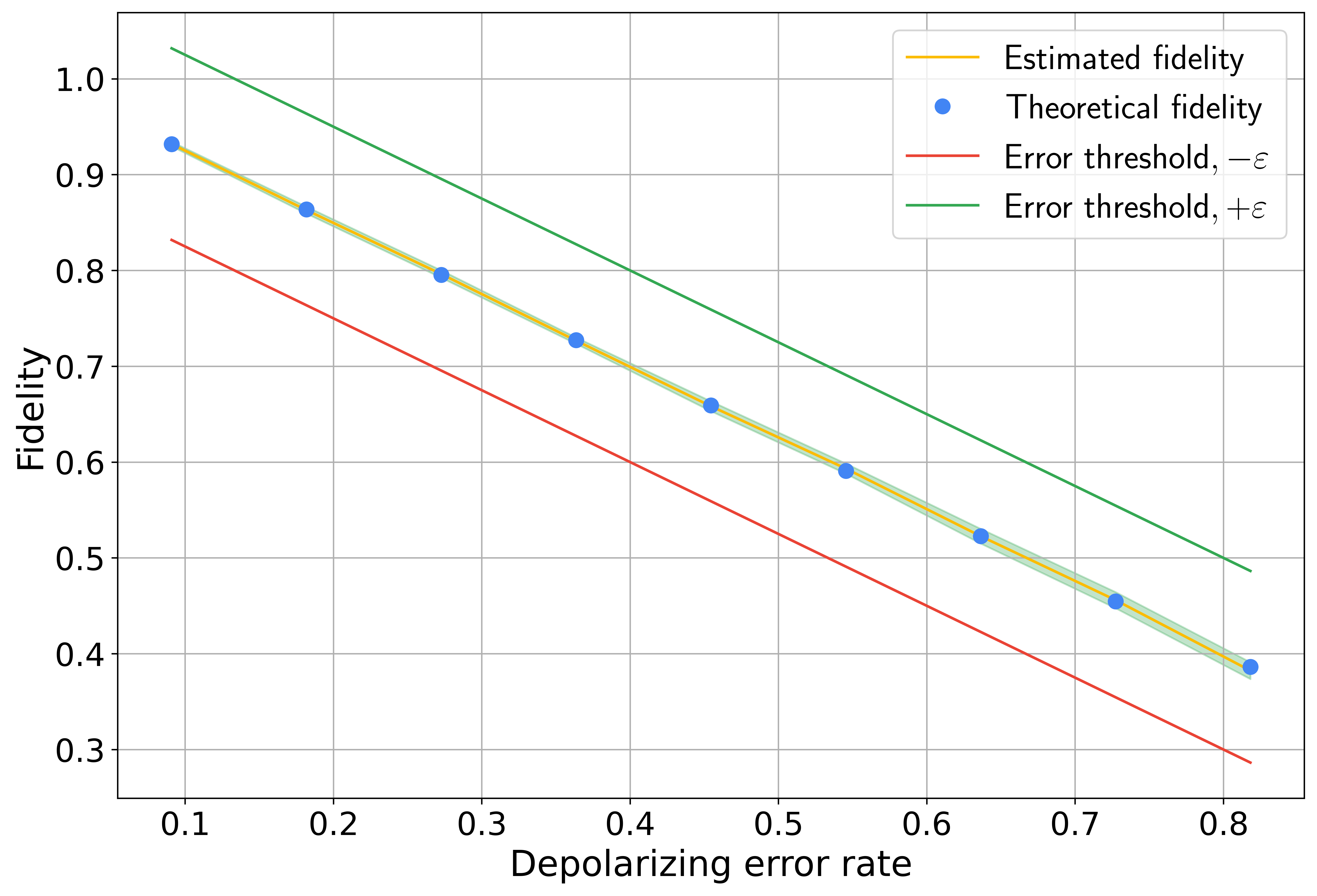}
	\caption{Apply \textbf{Protocol~\ref{ptl: efficient protocol}}
                to estimate the fidelity of a Bell measurements.
            The Bell measurement suffers from depolarizing error scenarios. The variance (shaded area) in spectra is obtained by repeating the efficient \textbf{Protocol \ref{ptl: efficient protocol}} $10$ times.}
	\label{fig1}
\end{figure}
\begin{figure}[!htbp]
	\centering
	\includegraphics[width=0.8\textwidth]{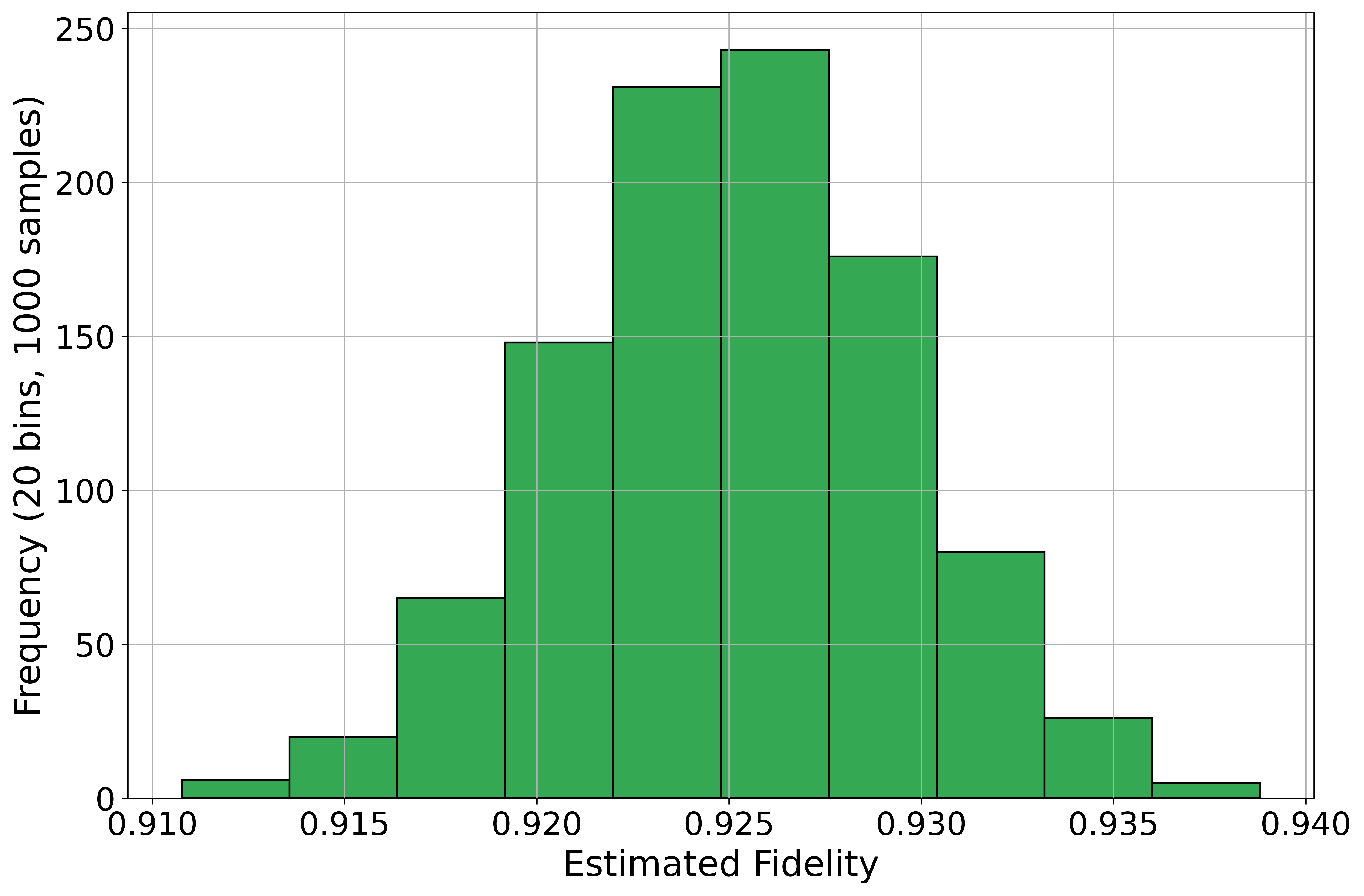}
	\caption{The histogram of the estimated measurement fidelities of a Bell measurements.}
	\label{fig: hist Bell Meas}
\end{figure}
\begin{figure}[!hbtp]
	\centering
	\includegraphics[width=0.8\textwidth]{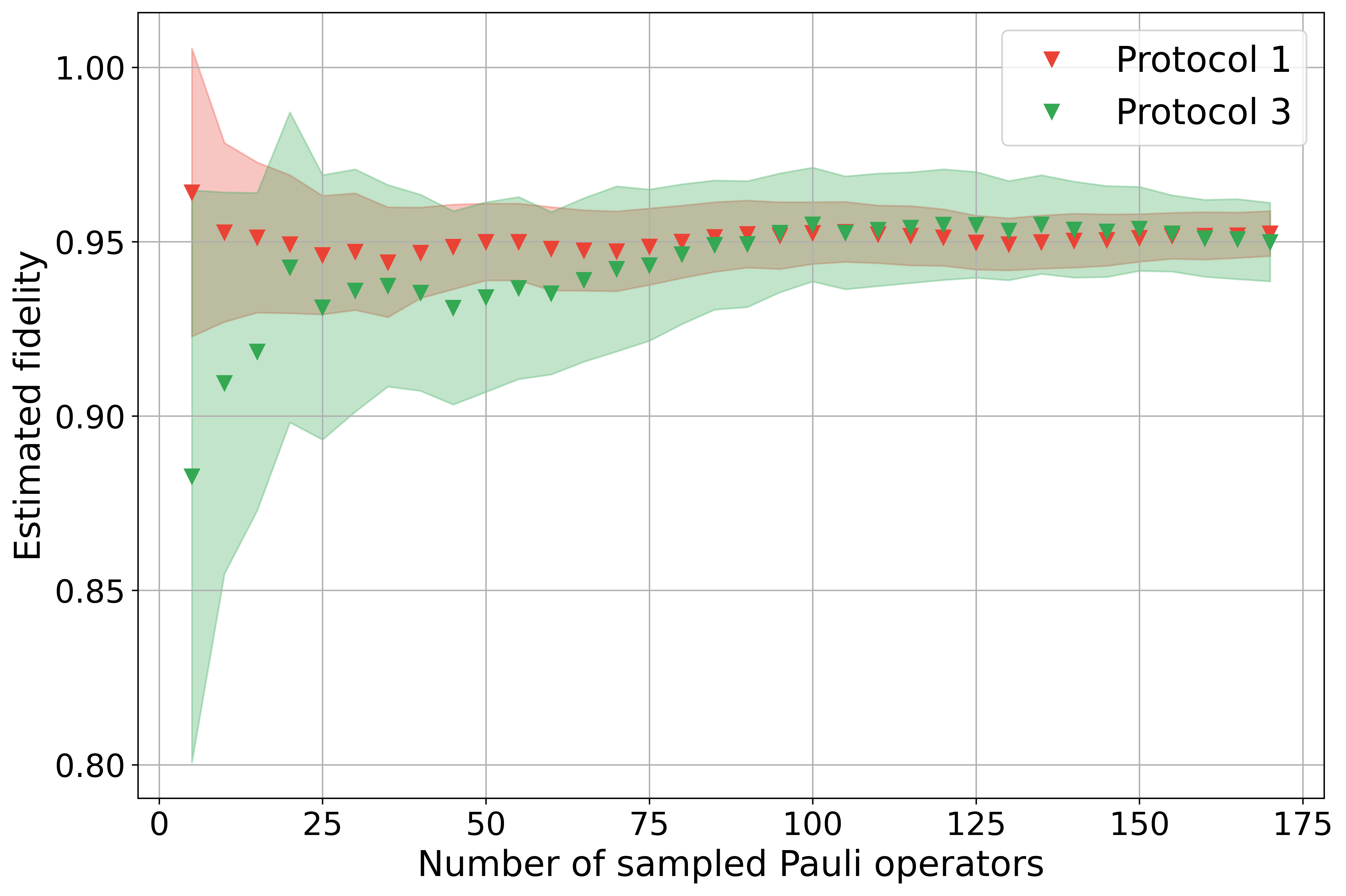}
	\caption{Performance comparisons of
            \textbf{Protocol~\ref{ptl: efficient protocol}} 
            and \textbf{Protocol~\ref{ptl: original protocol}}
            on estimating the EJM measurement with $\theta = \pi / 3$.
            The variance (shaded area) in spectra is obtained by repeating each protocol $30$ times.}
	\label{fig: comparisons be two protocol on EJM with 0.33pi}
\end{figure}

In an additional investigation, we scrutinize the variations in estimated fidelity 
under a constant depolarizing error rate $p=0.1$. 
As illustrated in Fig.~\ref{fig: hist Bell Meas}, 
a histogram plotting the distribution of $1000$ estimated fidelity values is 
visualized under the parameter conditions $\varepsilon = \delta = 0.05$. 
Notably, all the estimated fidelity values remain within 
a predefined additive error range of $2\varepsilon=0.1$ 
centered at the theoretical fidelity $0.925$. 
Consequently, the results portrayed in Figs.~\ref{fig1} and~\ref{fig: hist Bell Meas} 
furnish compelling numerical substantiation for the validity of our theoretical findings.

Expanding beyond the fidelity estimation of Bell measurements,
our investigation delves into the comprehensive evaluation of measurement fidelity 
for the recently introduced EJM schemes~\cite{gisin2019entanglement}. 
The generalized EJM bases, parameterized by $\theta$, are defined as~\cite{tavakoli2021bilocal}
\begin{align}
\label{def: ejm base}
   \vert \Phi_b^{\theta} \rangle
&= \frac{\sqrt{3} + \exp(i \theta)}{2\sqrt{2}} \vert \overrightarrow{m_b}, - \overrightarrow{m_b} \rangle 
 + \frac{\sqrt{3} - \exp(i \theta)}{2\sqrt{2}}\vert - \overrightarrow{m_b}, \overrightarrow{m_b} \rangle,
\end{align}
where $\vert \overrightarrow{m_b} \rangle$ represents the pure qubit states 
pointing towards the four vertices of a regular tetrahedron.
In Fig.~\ref{fig: comparisons be two protocol on EJM with 0.33pi}, 
we present the outcomes of our convergence analysis for the EJM schemes 
with a parameter setting of $\theta = \pi / 3$, 
employing both our efficient protocol and the direct protocol. 
To assess the variability in estimation, we repeat each protocol $30$ times.
The figures unmistakably demonstrate that the estimation outcomes for 
both protocols converge as the number of sampled Pauli operators increases. 
Notably, our efficient protocol exhibits significantly swifter convergence compared to the direct protocol.
Consequently, our efficient protocol not only enjoys a substantial advantage in terms of sample complexity, 
but also demonstrates superior or comparable convergence performance. 
Additional results from the convergence analysis can be 
found in Appendix~\ref{appx: more convergence analysis results}.

In Fig.~\ref{fig: histogram of the total number of samples for different theta of EJM}, 
we present a graphical representation detailing the total number of samples required across $1000$ 
uniformly sampled parameter values $\theta$ for the EJM measurement. 
The desired accuracy parameters were consistently set at $\varepsilon = 0.05$ and $\delta = 0.05$. 
The results are depicted in 
Fig.~\ref{fig: histogram of the total number of samples for different theta of EJM}, 
revealing that approximately $97.8\%$ of trials required fewer than five times the average total 
number of samples for successful estimation. 
This observation underscores the efficiency and robustness of our protocol 
across a diverse range of EJM parameter regions.

\begin{figure}[!hbtp]
	\centering
	\includegraphics[width=0.75\textwidth]{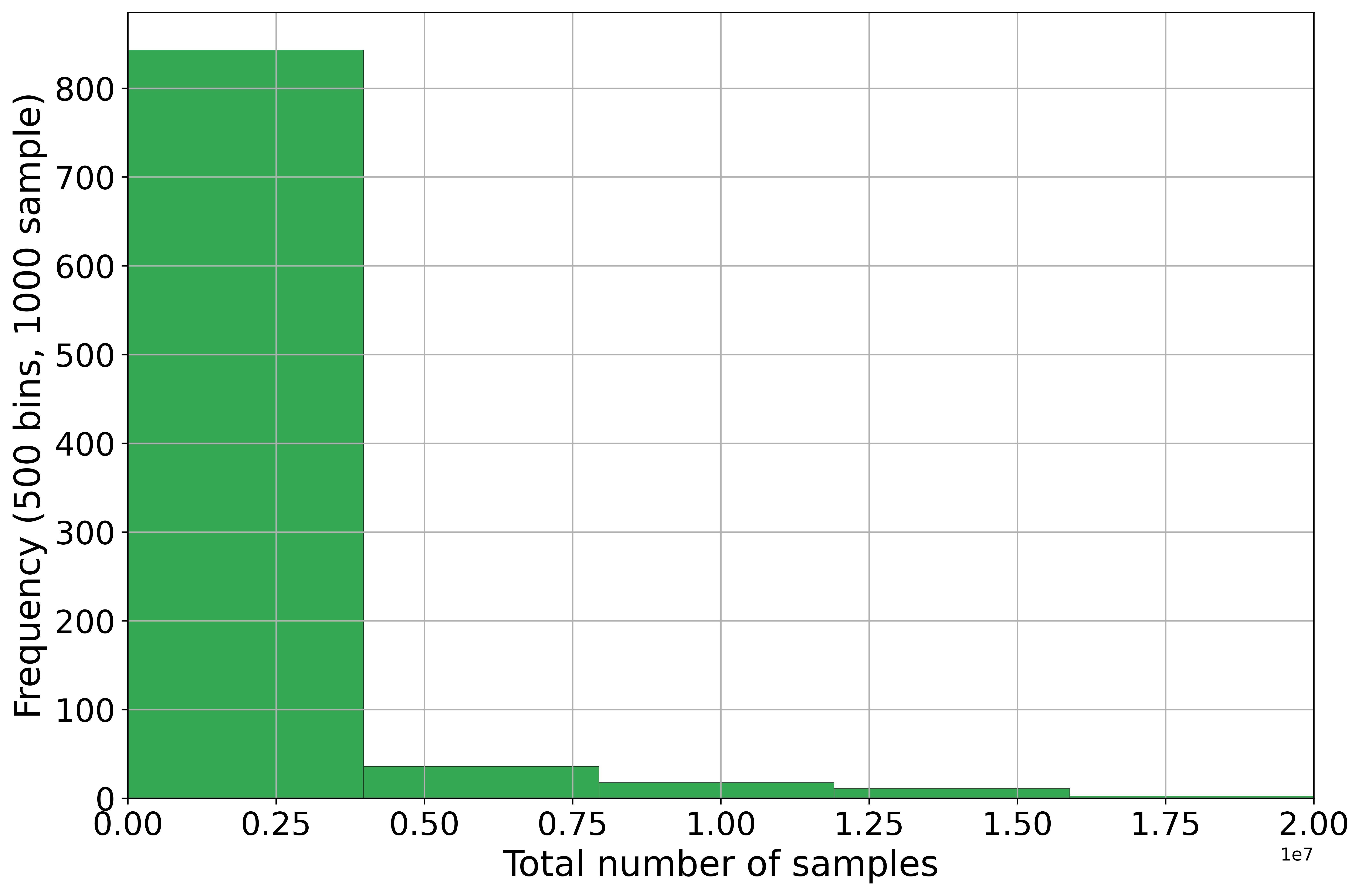}
	\caption{The histogram of the total number of samples for different $\theta$ of EJM measurement.}
	\label{fig: histogram of the total number of samples for different theta of EJM}
\end{figure}

To gain deeper insights into the hardness of the estimation process, we graphically represented the average sample complexity and the bounding results in terms of stabilizer R\'enyi entropy, as depicted in Fig. \ref{fig: magic bounds}. This visualization encompasses diverse parameter values of $\theta$ from Eq. \eqref{def: ejm base} in the context of EJM. Notably, the average sample complexity is 
consistently bounded by stabilizer R\'enyi entropies of the target state across all chosen EJM parameters. 
The results presented in Fig.~\ref{fig: magic bounds} also unveil that the sample complexities 
exhibit small variation with respect to the EJM parameters. 
Moreover, \textbf{Protocol \ref{ptl: efficient protocol}} achieves
a consistently reduced average sample complexity when juxtaposed with 
\textbf{Protocol \ref{ptl: original protocol}}.
However, the stabilizer R\'enyi entropy bounds of the direct \textbf{Protocol \ref{ptl: original protocol}} 
are much tighter than those of \textbf{Protocol \ref{ptl: efficient protocol}}, 
indicating that \textbf{Protocol \ref{ptl: efficient protocol}} might have a 
larger fluctuation in sample complexity.

\begin{figure}[!htbp]
	\centering
\includegraphics[width=0.8\textwidth]{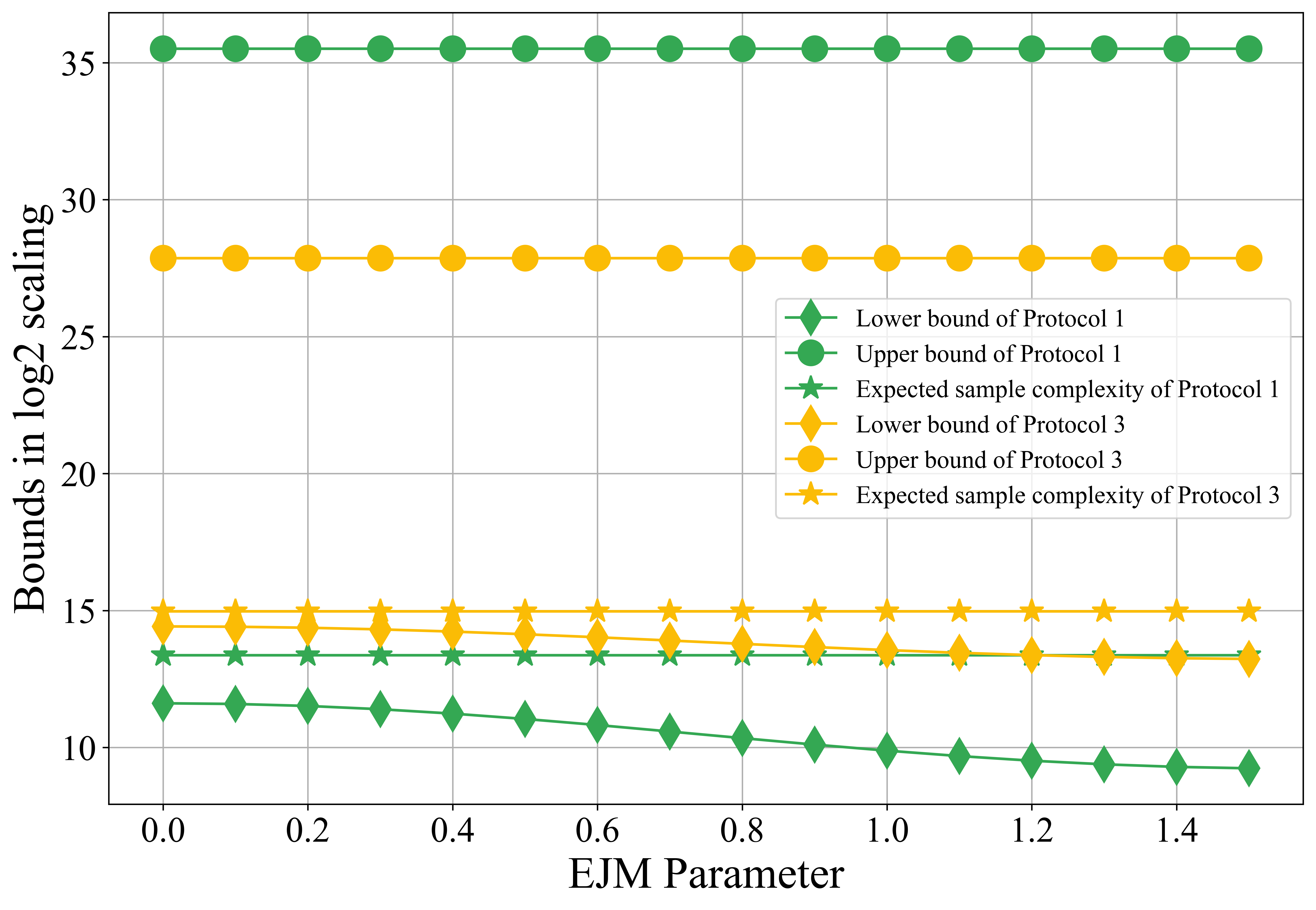}
	\caption{Comparisons of the expected sample complexity 
            and the stabilizer R\'enyi entropy bounds
            of \textbf{Protocol \ref{ptl: efficient protocol}} 
            and \textbf{Protocol \ref{ptl: original protocol}}.}
	\label{fig06}
 \label{fig: magic bounds}
\end{figure}

\section{Conclusions}

We have introduced an efficient protocol 
for estimating the performance of quantum measurements, 
leveraging the preparation and measurement of product eigenstates of Pauli operators. 
This protocol confers several advantages over existing methods, surpassing the speed of 
detector tomography and demonstrating an exponential acceleration compared to the 
direct fidelity estimation protocol for measurement devices. 
Furthermore, in specific scenarios, such as well-conditioned quantum measurements, 
our protocol demands only a constant number of measurement calls. 
This protocol is not only applicable to symmetric informationally complete POVMs 
and quantum measurements relying on mutually unbiased bases~\cite{renes2004symmetric,adamson2010improving,czartowski2020isoentangled}, 
but also the most general positive operator-valued measure measurements via Naimark's theorem. 
Its significance lies in its ability to facilitate the development of reliable quantum devices and 
propel continuous enhancements in their performances.
Our protocol opens up new possibilities 
for certifying quantum measurements to high accuracy, 
providing a valuable tool for advancing the field of quantum information processing 
and enabling the realization of reliable and high performance quantum measurement devices.

The quest for even more efficient fidelity estimation lies in navigating the intricacies of sophisticated sampling strategies and adaptive approaches. This path presents exciting challenges and opportunities to minimize variance and sample complexity, ultimately yielding sharper accuracy and broader applicability of quantum measurements.

\vspace{0.2in}
\textbf{Acknowledgements.}---This work was done when Z. S. was a research intern at Baidu Research.

\newcommand{\etalchar}[1]{$^{#1}$}




\setcounter{secnumdepth}{2}
\appendix
\widetext

\section{Proof of Theorem~\ref{thm: sc of efficient protocol}}\label{appx: the efficient protocol}

Before delving into the sample complexity analysis summarized in Theorem~\ref{thm: sc of efficient protocol}
(see also Algorithm \ref{ptl: efficient protocol}), 
we first present some essential lemmas.
\begin{lemma}
    \label{lemma: trace(AB)}
    Let $A$ and $B$ be two positive semidefinite operators, we have $\tr[AB] \geq 0$.
\end{lemma}
\begin{proof}
\label{pf: trace AB}
    By the spectral theorem, we have
    \begin{align}
        \tr[AB] 
        &= \tr\left[\sum_i \alpha_i \vert \phi_i \rangle \langle \phi_i \vert \sum_j \beta_j \vert \psi_j \rangle \langle \psi_j \vert \right] \\
        &= \sum_{ij} \alpha_i \beta_j \langle \phi_i \vert \psi_j \rangle \langle \psi_j \vert \phi_i \rangle \\
        &= \sum_{ij} \alpha_i \beta_j \vert  \langle \phi_i \vert \psi_j \rangle \vert^2 \\
        &\qquad\eqnote{Since $A,B$ are positive semidefinite, $\alpha_i \geq 0$ and $\beta_j \geq 0$} \nonumber \\
        &\geq 0.
    \end{align}
\end{proof}

\begin{lemma}
\label{lemma: trace(Vk square)}
    Let $\left\{V_k\right\}_k$ be a set of POVM elements satisfying $\sum_{k=1}^{2^n} V_k= \id$, where $n$ is the number of qubits. We then have
    \begin{align}
        \sum_{k=1}^{2^n} \tr[V_k^2]
        &\leq 2^n.
    \end{align}
\end{lemma}
\begin{proof}
\label{pf: trace Vk2}
\begin{align}
\tr\left[\sum_{k=1}^{2^n} V_k^2\right] 
    &= \tr\left[\sum_{k=1}^{2^n} V_k \sum_{j=1}^{2^n} V_j\right] -\tr\left[\sum_{k \neq j} V_k V_j\right] \\
    &= \tr[\id] - \tr\left[\sum_{k \neq j} V_k V_j \right] \\
    &\qquad\eqnote{By lemma (\ref{lemma: trace(AB)}), we have $\sum_{k \neq j} \tr[V_k V_j] \geq 0$} \nonumber \\
    &\leq 2^n. 
\end{align}
\end{proof}

\begin{proof}[Proof of Theorem \ref{thm: sc of efficient protocol}]
To commence, we begin with bounding the estimation error from the sample mean in step (11) of the efficient \textbf{Protocol \ref{ptl: efficient protocol}}. Recognizing that $X_{l}$ is an unbounded random variable, we employ Chebyshev's inequality to derive a tail bound for the estimation $Y$ in Eq.~\eqref{ey}. Employing the definitions of the $q_{l}$ and $X_{l}$, we analyze the variance of $X_{l}$ as
\begin{align}
    {\rm Var}[X_l] 
&= \sum_{l=1}^{4^n} \left(\sum_{k'=1}^{2^n}\frac{\tr[\psi_{k’}P_l]^2}{2^n}\right)\times 
   \left(\sum_{k=1}^{2^n}\tr[V_kP_l]\frac{\tr[\psi_kP_l]}{\sum_{k'=1}^{2^n}\tr[\psi_{k’}P_l]^2}\right)^2 - F^2 \\
&\leq \sum_{l=1}^{4^n} \frac{\sum_{k'=1}^{2^n}\tr[\psi_{k’}P_l]^2}{2^n}\times
        \frac{\left[\sum_{k=1}^{2^n}\tr[V_kP_l]\tr[\psi_kP_l]\right]^2}{\left[\sum_{k'=1}^{2^n}\tr[\psi_{k’}P_l]^2\right]^2} \\
&= \sum_{l=1}^{4^n} \frac{\left[\sum_{k=1}^{2^n}\tr[V_kP_l]\tr[\psi_kP_l]\right]^2}{2^n\left(\sum_{k'=1}^{2^n}\tr[\psi_{k’}P_l]^2\right)} \\
&\qquad\eqnote{Cauchy–Schwarz inequality}\nonumber\\
&\leq \sum_{l=1}^{4^n} \frac{\left(\sum_{k=1}^{2^n}\tr[V_kP_l]^2\right)\left(\sum_{k=1}^{2^n}\tr[\psi_kP_l]^2\right)}{2^n\left(\sum_{k'=1}^{2^n}\tr[\psi_{k’}P_l]^2\right)} \\
&= \frac{\sum_{l=1}^{4^n} \sum_{k=1}^{2^n}\tr[V_kP_l]^2}{2^n} \\
&= \frac{\sum_{k=1}^{2^n} \tr[V_k^2]}{2^n} \\
&\qquad\eqnote{By lemma (\ref{lemma: trace(Vk square)})} \\
&\leq 1.
\end{align}
Hence, it is  adequate to sample $m=\lceil1/(\varepsilon^2\delta)\rceil$ number of Pauli operators, and the estimation of $\cF$ can be expressed as
\begin{align}
    Y = \frac{1}{m}\sum_{i=1}^m X_{l_i}.
\end{align}
This formulation results in the $\varepsilon$-close estimator in Eq.~\eqref{estimation step 1} as guaranteed by Chebyshev's inequality.

Next, we bound the estimation error of $X_{l_i}$ in step (9) of the efficient \textbf{Protocol \ref{ptl: efficient protocol}}. We show that $\wh{X}_{a_j}^{(l_i)}$ is an unbiased estimator for a fixed $l_i$, since 
\begin{align}
    \bE[\wh{X}_{a_j}^{(l_i)}] 
    &= \bE\left[\frac{2^n}{\sum_{k'=1}^{2^n}\tr[\psi_{k’}P_{l_i}]^2}\times\lambda^{(l_i)}_{a_j}\times \delta_{k o_j} \tr[\psi_{o_j}P_{l_i}]\right] \\
    &= \bE\left[\frac{2^n}{\sum_{k'=1}^{2^n}\tr[\psi_{k’}P_{l_i}]^2}\times\lambda^{(l_i)}_{a_j}\times \sum_{k=1}^{2^n} \delta_{k o_j} \tr[\psi_{k}P_{l_i}]\right] \\
    &= \bE_{a_j}\left[\frac{2^n}{\sum_{k'=1}^{2^n}\tr[\psi_{k’}P_{l_i}]^2}\times\lambda^{(l_i)}_{a_j}\times \sum_{k=1}^{2^n} \bE[\delta_{k o_j}] \tr[\psi_{k}P_{l_i}]\right]\\
    &= \bE_{a_j}\left[\frac{2^n}{\sum_{k'=1}^{2^n}\tr[\psi_{k’}P_{l_i}]^2}\times\lambda^{(l_i)}_{a_j}\times \sum_{k=1}^{2^n} \tr[V_{k} \phi_{a_j}^{(l_i)}]\tr[\psi_{k}P_{l_i}]\right] \\ 
    &=\sum_{a_j=1}^{2^n} \frac{1}{2^n} \times \frac{2^n}{\sum_{k'=1}^{2^n}\tr[\psi_{k’}P_{l_i}]^2}\times\lambda^{(l_i)}_{a_j} \sum_{k=1}^{2^n} \tr[V_{k} \phi_{a_j}^{(l_i)}]\tr[\psi_{k}P_{l_i}] \\ 
    &= \frac{1}{\sum_{k'=1}^{2^n}\tr[\psi_{k’}P_{l_i}]^2} \sum_{k=1}^{2^n} \tr\left[V_{k} P_{l_i} \right]\tr[\psi_{k}P_{l_i}] \\
    &\equiv X_{l_i},
\end{align}
where $\delta_{ko_j}=1$ if $k=o_j$ and $\delta_{ko_j}=0$ otherwise. Run step (6) and step (7) in the efficient \textbf{Protocol \ref{ptl: efficient protocol}} a total number of $n_{l_i}$ times, the estimation of $X_{l_i}$ is
\begin{align}
    \wt{X}_{l_i}
    &= \frac{1}{n_{l_i}} \sum_{j=1}^{n_{l_i}} \wh{X}_{a_j}^{(l_i)} = \frac{1}{n_{l_i}} \sum_{j=1}^{n_{l_i}} \frac{2^n}{\sum_{k'=1}^{2^n}\tr[\psi_{k’}P_{l_i}]^2}\times\lambda^{(l_i)}_{a_j}\times\delta_{k o_j}\times\tr[\psi_{o_j}P_{l_i}].
\end{align}
Then, we consider the double sum
\begin{align}
    m \wt{Y} 
= \sum_{i=1}^m \wt{X}_{l_i}
= \sum_{i=1}^m\sum_{j=1}^{n_{l_i}}
    \underbrace{\frac{1}{n_{l_i}}\frac{2^n}{\sum_{k'=1}^{2^n}\tr[\psi_{k’}P_{l_i}]^2}\times\lambda^{(l_i)}_{a_j}\times\delta_{k o_j}\times\tr[\psi_{o_j}P_{l_i}]}_{C_{ij}}.
\end{align}
Since $\lambda^{(l_i)}_{a_j}=\pm1/\sqrt{2^n}$, $\delta_{ko_j}=0$ when $k \neq o_j$ and $1$ otherwise, and {$\vert \tr[\psi_{o_j}P_{l_i}]\vert \leq 1/\sqrt{2^n}$}, for a fixed $l_i$, the upper bound $b_{l_i}$ and the lower bound $a_{l_i}$ of $C_{ij}$ is
\begin{align}
   b_{l_i}= -a_{l_i} = \frac{1}{n_{l_i}}\frac{1}{\sum_{k'=1}^{2^n}\tr[\psi_{k’}P_{l_i}]^2}.
\end{align}
We use the Hoeffding's inequality to bound the estimation error as:
\begin{align}
    {\rm Pr}(\vert\wt{Y}-Y\vert\geq\varepsilon) = {\rm Pr}(\vert m \wt{Y} - mY\vert\geq\varepsilon)\leq 2\exp(-2 m^2\varepsilon^2/C),
\end{align}
where $C = \sum_{i=1}^m\sum_{j=1}^{n_{l_i}}(2/(n_{l_i} \sum_{k^\prime=1}^{2^n} \tr[\psi_{k^\prime}P_{l_i}]^2))^2$. Now choose
\begin{align}
    n_{l_i} 
= \left\lceil\frac{2}{m\left(\sum_{k'=1}^{2^n}\tr[\psi_{k’}P_{l_i}]^2\right)^2\varepsilon^2}\ln\frac{2}{\delta} \right\rceil
= \left\lceil\frac{1}{(\sum_{k'=1}^{2^n}\tr[\psi_{k’}P_{l_i}]^2)^2}\times\frac{2}{m\varepsilon^2}\ln\frac{2}{\delta} \right\rceil.
\end{align}
We have
\begin{align}
\frac{1}{m^2} C &= \sum_{i=1}^m\sum_{j=1}^{n_{l_i}} 4\frac{1}{m^2n_{l_i}^2}
            \frac{1}{(\sum_{k'=1}^{2^n}\tr[\psi_{k’}P_{l_i}]^2)^2} \\
  &= \sum_{l=1}^m 4 \frac{1}{m^2n_{l_i}}\frac{1}{(\sum_{k'=1}^{2^n}\tr[\psi_{k’}P_{l_i}]^2)^2} \\
  &\leq \sum_{i=1}^m \frac{1}{m} \frac{2\varepsilon^2}{\ln(2/\delta)} \\
  &= \frac{2\varepsilon^2}{\ln(2/\delta)},
\end{align}
which guarantees the confidence interval \eqref{estimation step 2}. 
By combination of Eqs.~\eqref{estimation step 1} and \eqref{estimation step 2}, 
we arrive at the final estimator $\wt{Y}$ satisfying the confidence 
interval~\eqref{est: eps-close est for efficient}.

To determine the expected number of times the measurement device needs to be accessed, 
it is crucial to acknowledge that $n_{l}$ is a random variable, as $l$ is selected randomly. 
By the definition of sampling, for a fixed $l$, we have
\begin{align}
   \bE[n_l] 
&= \sum_{l=1}^{\vert \cS \vert}q_l n_l \\
&\geq \sum_{l=1}^{\vert \cS \vert}\frac{\sum_{k'=1}^{2^n}\tr[\psi_{k’}P_l]^2}{2^n}\times
        \left(\frac{2}{m\left(\sum_{k'=1}^{2^n}\tr[\psi_{k’}P_l]^2\right)^2\varepsilon^2}
        \log\frac{2}{\delta}\right) \\
&=  \left(\sum_{l=1}^{\vert \cS \vert}\frac{1}{2^n 
    \sum_{k'=1}^{2^n}\tr[\psi_{k’}P_l]^2}\right)\times\frac{2}{m\varepsilon^2}\log\frac{2}{\delta}.
\end{align}
Due to the HM-GM-AM-QM inequalities, we have
\begin{align}
    {\left(\sum_{l=1}^{\vert \cS \vert}\frac{1}{2^n \sum_{k'=1}^{2^n}\tr[\psi_{k’}P_l]^2}\right)}
\geq \frac{\vert \cS \vert^2}{2^n \sum_{l=1}^{\vert \cS \vert}{\sum_{k'=1}^{2^n}\tr[\psi_{k’}P_l]^2}}
= \frac{\vert \cS \vert^2}{4^n}.
\end{align}
Thus, we get
\begin{align}
    \mathbb{E}[L] = \sum_{i=1}^m \bE[n_{l_i}]
    \geq {\frac{\vert \cS \vert^2}{4^n} }\frac{2}{\varepsilon^2}\ln\frac{2}{\delta}.
\end{align}
It is noteworthy that the maximum cardinality of $\vert \cS \vert$ is $4^n$, 
signifying that the worst-case lower bound scales as $\Omega(4^n)$. 
\end{proof}

\section{Proof of Theorem~\ref{thm: magic bounds for new method}}
\label{appx: magic bounds for efficient protocol}

\begin{proof}[Proof of Theorem~\ref{thm: magic bounds for new method}]
Recall the estimator \eqref{wtY-ptl-1} employed in the efficient protocol. Setting $n_{l_i} = 1$, we obtain
\begin{align}
    \wt{Y} 
    &= \frac{1}{L} \sum_{i=1}^L
    \frac{2^n}{\sum_{k'=1}^{2^n}\tr[\psi_{k’}P_{l_i}]^2}\times\lambda^{(l_i)}_{a_j} \times \delta_{ko_j}\times\tr[\psi_{o_j}P_{l_i}] \equiv \frac{1}{L} \sum_{i=1}^{L} C_{l_i}.
\end{align}
The upper bound $b_{l_i}$ and the lower bound $a_{l_i}$ of $C_{l_i}$ satisfy $b_{l_i} = - a_{l_i} = 1/\sum_{k^{\prime}=1}^{2^n} \tr[\psi_{k^{\prime}} P_{l_i}]^2$. 
Then, applying Hoeffding's inequality, we obtain
\begin{align}
    {\rm Pr}\left(\left\vert \wt{Y} - \cF \right\vert \geq \varepsilon\right)
    &= {\rm Pr}\left(\left\vert \sum_{i=1}^L C_{l_i} - L \cF \right\vert \geq L \varepsilon\right) \\
    &\leq 2\exp\left[- \frac{L^2 \varepsilon^2}{\sum_{i=1}^L \frac{2}{\left[\sum_{k=1}^{2^n} \tr[\psi_k P_{l_i}]^2\right]^2}}\right] \\
    &\leq 2 \exp\left[-\frac{L^2 \varepsilon^2}{\sum_{i=1}^L \frac{2}{\mathbb{M}^2}}\right] \\
    &\leq \delta,
\end{align}
where 
\begin{align}
    \mathbb{M}
    &:= \min_l \sum_{k=1}^{2^n} \tr[\psi_k P_{l}]^2
    = \min_l  \tr[O P_{l}]^2 \cdot 2^n 
    \leq \sqrt{\frac{\exp[-M_2(\cM)]}{2^n}}^2 \cdot 2^n
    ={\exp\left[- M_2(\cM)\right]}
\end{align}
by the definition of $O$ in Eq.~\eqref{def: o} and $M_{\alpha}(\cM)$ in Eq.~\eqref{def: alpha sre}. Thus, we get the lower bound as
\begin{align}
    L
    \geq \frac{1}{\mathbb{M}^2} \cdot \frac{2}{\varepsilon^2} \ln\frac{2}{\delta}
    \geq  \frac{2}{\varepsilon^2} \ln\frac{2}{\delta} \exp\left[2 M_2(\cM)\right].
\end{align}

Now, let's analyze the upper bound. Initially, we truncate the PVM elements $\psi_k$ as defined in Eq.~\eqref{def: truncated psi_k}, resulting in $\psi_k^{\prime}$. Subsequently, we employ these truncated PVM elements $\psi_k^\prime$ to approximate $\cF$. The absolute difference is bounded as $\vert \cF' - \cF \vert < {\varepsilon}/{2}$. The estimation $\wt{Y}^{\prime}$ of $\cF^{\prime}$ is given by
\begin{align}
    \wt{Y}^{\prime} 
    &= \frac{1}{L} \sum_{i=1}^L
    \frac{2^n}{\sum_{k'=1}^{2^n}\tr[\psi_{k^{\prime}}^{\prime}P_{l_i}]^2}\times\lambda^{(l_i)}_{a_i} \times \delta_{ko_i}\times\tr[\psi_{o_i}P_{l_i}] \equiv \frac{1}{L} \sum_{i=1}^{L} C_{l_i}^{\prime}.
\end{align}
The upper bound $b_{l_i}^\prime$ and the lower bound $a_{l_i}^\prime$ of $C_{l_i}^\prime$ satisfy $b_{l_i}^\prime = - a_{l_i}^\prime = 1/\sum_{k^{\prime}=1}^{2^n} \tr[\psi_{k^{\prime}}^{\prime} P_{l_i}]^2$.
It can be demonstrated that $\bE[C_{l_i}^\prime]=\cF$.
Subsequently, by applying Hoeffding's inequality, we obtain 
\begin{align}
    {\rm Pr}\left(\vert \wt{Y}^{\prime} - \cF^{\prime} \vert \leq \frac{\varepsilon}{2}\right)
&= {\rm Pr}\left(\left\vert  \sum_{i=1}^L C_{l_i}^{\prime} - L \cF^{\prime} \right\vert 
    \geq \frac{L \varepsilon}{2}\right) \\
&\geq 1 - 2 \exp\left[-\frac{L^2 \varepsilon^2/4}{\sum_{i=1}^L (b_{l_i}^\prime - a_{l_i}^\prime)^2}\right] \\
&\geq 1 -  2\exp\left[-\frac{L^2 \varepsilon^2}{\sum_{i=1}^L 
    \frac{16}{\mathbb{M}^{\prime 2}}}\right] \\
&= 1 - \delta,
\end{align}
where $\mathbb{M}^{\prime}:= \min_l \sum_{k=1}^{2^n} \tr[\psi_k^\prime P_l]^2$. This leads to
\begin{align}
    L = \frac{16 }{\mathbb{M}^{\prime 2}} \cdot \frac{1}{\varepsilon^2} \ln\frac{2}{\delta}.
\end{align}

Now, we derive:
\begin{align}
    \mathbb{M}^{\prime}
    = \min_{l} \sum_{k=1}^{2^n} \tr[\psi_k^\prime P_l]^2 
    &\geq \min_{l} \sum_{k=1}^{2^n} \tr[\psi_{k, cut} P_l]^2 \\
    &\geq \min_{l} \sum_{k=1}^{2^n} \left[\frac{\varepsilon}{2\sqrt{2}} 
            \sqrt{\frac{\exp[-M_0(\psi_k)]}{2^n}}\right]^2 \\
    &= \sum_{k=1}^{2^n} \frac{\varepsilon^2}{8} \frac{\exp[-M_0(\psi_k)]}{2^n} \\
    &\geq \sum_{k=1}^{2^n} \frac{\varepsilon^2}{8} \frac{\exp[-M_0(\cM)]}{2^n} \\
    &=\frac{\varepsilon^2}{8} {\exp[- M_0(\cM)]}, 
\end{align}
Finally, the upper bound is obtained as
\begin{align}
   L = \frac{16 }{\mathbb{M}^{\prime 2}} \cdot \frac{1}{\varepsilon^2} \ln\frac{2}{\delta} \leq \frac{2^{10}}{\varepsilon^6}\ln\frac{2}{\delta} \exp[2M_0(\cM)],
\end{align}
Substituting $\varepsilon$ and $\delta$ with $2 \varepsilon$ and $2 \delta$, respectively, yields Theorem \ref{thm: magic bounds for new method}, presenting the bounding results for a $2 \varepsilon$-close estimator.
\end{proof}

\section{Proof of Theorem~\ref{thm: optimal lower bound for measurement fidelity}}
\label{appx: naive protocol}

\begin{proof}[Proof of Theorem~\ref{thm: optimal lower bound for measurement fidelity}]
It follows directly from Eq.~\eqref{eq:global-fidelity} that $\bE[\wt{X}_{k_i}]=\cF$.
We also note that $\wt{X}_{k_i}\in\{0,1\}$. 
Thus, by Hoeffding's inequality, we get
\begin{align}
{\rm Pr}\left(\vert \wt{Y} - \cF \vert \geq \varepsilon \right)
&= {\rm Pr}\left(\left\vert \sum_{i=1}^{L} \wt{X}_{k_i} - L \cF \right\vert \geq L \varepsilon \right) \\
&\leq 2 \exp\left[-\frac{2 L^2 \varepsilon^2}{\sum_{i=1}^L (1 - 0)^2}\right] \\
&= 2 \exp\left[-2 L \varepsilon^2 \right] \\
&\leq \delta,
\end{align}
which implies that 
\begin{align}
L \geq \frac{1}{2 \varepsilon^2} \ln\frac{2}{\delta}.
\end{align}
This leads to a $\varepsilon$-close estimator 
\begin{align}
  {\rm Pr}\left(\vert \wt{Y} - \cF \vert \geq \varepsilon \right) \leq \delta.
\end{align}
By substituting $\varepsilon$ and $\delta$ with $2 \varepsilon$ and $2 \delta$, respectively, 
we conclude Theorem \ref{thm: optimal lower bound for measurement fidelity}.
\end{proof}

\section{Details of the direct protocol}
\label{appx: the original protocol}

Now we proceed to formulate the direct protocol. First, we rewrite the measurement fidelity $\cF$ as
\begin{align}
    \cF
    &= \frac{1}{2^n} \sum_{k=1}^{2^n} \tr[\psi_k V_k] \\
    &= \frac{1}{2^n} \sum_{k=1}^{2^n} \sum_{l=1}^{4^n} \tr[\psi_k P_{l}] \tr[V_k P_{l}] \\
    &= \sum_{k=1}^{2^n} \sum_{l=1}^{4^n} \frac{\tr[\psi_k P_{l}]^2}{2^n} \frac{\tr[V_k P_{l}]}{\tr[\psi_k P_{l}]} \\
    &\equiv \sum_{k=1}^{2^n} \sum_{l=1}^{4^n} q_{kl} X_{kl},
\end{align}
where $q_{kl}:= {\tr[\psi_k P_{l}]^2}/{2^n}$ forms a joint probability mass function and $X_{kl}:= {\tr[V_k P_l]}/{\tr[\psi_k P_{l}]}$ denotes a random variable. It can be demonstrated that $\bE[X_{kl}] = \cF$, indicating that $X_{kl}$ serves as an unbiased estimator for $\cF$. 
Consequently, we can formulate an estimator for $\cF$ utilizing the sample mean. We perform joint sampling $k$ and $l$ multiple times according to the probability distribution $q_{kl}$, resulting in $m$ pairs $(k_1, l_1), (k_2, l_2), ..., (k_m, l_m)$. We choose $m$ as
\begin{align}
    m = \left\lceil 1/(\varepsilon^2 \delta) \right\rceil.
    \label{m}
\end{align}
Then, we estimate $\cF$ using the following formula:
\begin{align}
\label{eq: y}
    Y = \frac{1}{m}\sum_{i=1}^m X_{k_i l_i}.
\end{align}
This leads to a $\varepsilon$-accurate estimator as
\begin{align}
    {\rm Pr}\left(\vert Y - \cF \vert \geq \varepsilon \right) \leq \delta.
    \label{oest1}
\end{align} 

Subsequently, we proceed to estimate $X_{k_i l_i}$ using the measurement outcomes. 
Given that $P_{l}$ is not a quantum state, we must execute the spectral decomposition
\begin{align}
P_{l} 
= \sum_{a=1}^{2^n}\lambda_{a}^{(l)} {\vert \phi_{a}^{(l)} \rangle \langle \phi_{a}^{(l)} \vert} 
\equiv \sum_{a=1}^{2^n}\lambda_{a}^{(l)} \phi_{a}^{(l)}.
\end{align}
Rewrite $X_{k_i l_i}$ as 
\begin{align}
    X_{k_i l_i}
    &= \sum_{a=1}^{2^n}\lambda_{a}^{(l_i)} \frac{\tr[V_{k_i} \phi_{a}^{(l_i)}]}{\tr[\psi_{k_i} P_{l_i}]}.
\end{align}
Next, $\tr[V_{k_i} \phi_{a}^{(l_i)}]$ can be estimated 
by performing measurements on $\phi_{a}^{(l_i)}$ multiple times. 
The subroutine for this process is outlined below: 
\begin{itemize}
    \item \textit{Sample an index.} Choose a random index $a_j$, where $a_j \in \{1,2,\cdots,2^n\}$ uniformly;
    \item \textit{Measure the eigenstate.}  Utilize the quantum measurement device to measure the eigenstate $\phi^{(l_i)}_{a_j}$ (associated with $P_{l_i}$) and record the experimental outcome $o_j$.
\end{itemize}

Define a new random variable $\wh{X}_{k_i l_i}$ as
\begin{align}
    \wh{X}_{k_i l_i} = \frac{2^n}{\tr[\psi_{k_i} P_{l_i}]} \times \lambda_{a_j}^{(l_i)} \times \delta_{k_i o_j},
\end{align}
and it can be demonstrated that $\bE[\wh{X}_{k_i l_i}]$ is an unbiased estimator as follows:
\begin{align}
    \bE[\widehat{X}_{k_i l_i}] 
    &= \bE\left[\frac{2^n}{\tr[\psi_{k_i} P_{l_i}]} \times \lambda_{a_j}^{(l_i)} \times \delta_{k_i o_j}\right] \\
    &= \bE_{a_j}\left[ \frac{2^n}{\tr[\psi_{k_i} P_{l_i}]} \times \lambda_{a_j}^{(l_i)} \times \bE[\delta_{k_i o_j}]\right]\\
    &= \bE_{a_j}\left[ \frac{2^n}{\tr[\psi_{k_i} P_{l_i}]} \times \lambda_{a_j}^{(l_i)} \times \tr[V_{k_i} \phi_{a_j}^{(l_i)}]\right] \\ 
    &=\sum_{a_j=1}^{2^n} \frac{1}{2^n} \times \frac{2^n}{\tr[\psi_{k_i} P_{l_i}]}  \lambda_{a_j}^{(l_i)} \tr[V_{k_i} \phi_{a_j}^{(l_i)}] \\ 
    &=  \frac{\tr[V_{k_i} P_{l_i}]}{\tr[\psi_{k_i} P_{l_i}]} \equiv X_{k_i l_i}
\end{align}
By running this procedure a total number of $n_{l_i}$ times, we can estimate $X_{k_i l_i}$ as
\begin{align}
\label{eq: wtX k_i l_i}
    \wt{X}_{k_i l_i} = \frac{1}{n_{l_i}} \sum_{j=1}^{n_{l_i}} \wh{X}_{k_i l_i}  
    =\frac{1}{n_i} \sum_{j=1}^{n_i} \frac{2^n}{\tr[\psi_{k_i} P_{l_i}]} \times \lambda_{a_j}^{(l_i)} \times \delta_{k_i o_j},
\end{align}
where Our estimator of $Y$ is 
\begin{align}
\label{eq: ptl o: wtY}
    \wt{Y}
    &= \frac{1}{m} \sum_{i=1}^m \wt{X}_{k_i l_i} \\
    &= \sum_{i=1}^{m} \sum_{j=1}^{n_{l_i}} \frac{1}{m n_{l_i}} \frac{2^n}{\tr[\psi_{k_i} P_{l_i}]} \times \lambda_{a_j}^{(l_i)} \times \delta_{k_i o_j},
\end{align}
where $\delta_{k_io_j}=1$ if $k_i=o_j$ and $\delta_{k_io_j}=0$ otherwise. To satisfy the confidence interval
\begin{align}
    {\rm Pr}\left(\vert \wt{Y} - Y \vert \geq \varepsilon \right) \leq \delta,
    \label{oest2}
\end{align}
we choose
\begin{align}
    n_{l_i}
    &= \left\lceil \frac{1}{\tr[\psi_{k_i} P_{l_i}]^2} \frac{2 \cdot 2^n}{ m \varepsilon^2} \ln{\frac{2}{\delta}}\right\rceil,
\end{align}

By combining the confidence intervals \eqref{oest1} and \eqref{oest2} using the union bound, we deduce the confidence interval for the final estimation $\wt{Y}$ as
\begin{align}
    {\rm Pr}\left\{\vert \wt{Y} - \cF \vert \geq 2 \varepsilon \right\} \leq 2 \delta.
    \label{est: eps-close est for original}
\end{align}

\section{Proof of Theorem~\ref{thm: sc of original protocol}}
\label{appx:sc of original protocol}

\begin{proof}[Proof of Theorem~\ref{thm: sc of original protocol}]
\label{pf: sample lower bound in pauli basis}
First, we bound the estimation error from the sample mean in step (11) of the 
direct \textbf{Protocol \ref{ptl: original protocol}}. 
We note that $X_{kl}$ is an unbounded random variable; thus we employ Chebyshev's inequality \cite{kliesch2021theory} to derive a tail bound for the estimation $Y$ in Eq.~\eqref{eq: y}. 
Using the definitions of the $q_{kl}$ and $X_{kl}$, we investigate the variance of $X_{kl}$ and obtain
\begin{align}
    {\rm Var}(X_{k l})
    &= \mathbb{E}\left[X_{k l}^2\right] - \mathbb{E}\left[X_{k l}\right]^2 \\
    &= \sum_{k=1}^{2^n} \sum_{l=1}^{4^n} \frac{\tr[\psi_{k} P_{l}]^2}{2^n} \times \frac{\tr[V_{k} P_{l}]^2}{\tr[\psi_{k}  P_{l}]^2}
    - \left[\sum_{k=1}^{2^n} \sum_{l=1}^{4^n} \frac{\tr[\psi_{k} P_{l}]^2}{2^n} \times \frac{\tr[V_{k} P_{l}]}{\tr[\psi_{k} P_{l}]} \right]^2 \\
    &= \sum_{k=1}^{2^n} \sum_{l=1}^{4^n}  \frac{\tr[V_{k} P_{l}]^2}{2^n} - \cF^2 \\
    &\leq \sum_{k=1}^{2^n} \frac{\sum_{l=1}^{4^n} \tr[V_{k} P_{l}]^2}{2^n} \\
    &= \sum_{k=1}^{2^n} \frac{\tr[V_{k}^2]}{2^n}\\
    &\qquad\eqnote{By lemma (\ref{lemma: trace(Vk square)})} \nonumber \\
    &\leq 1.
\end{align}
Therefore, we can bound the the variance of $Y$ as
\begin{align}
    {\rm Var}(Y)
    &= \frac{1}{m^2} \sum_{i=1}^m {\rm Var}(X_{k_i l_i}) \in \left[0, \frac{1}{m}\right].
\end{align}
Then, by Chebyshev's inequality, we have
\begin{align}
    {\rm Pr}\left(\vert Y - \cF \vert \geq \varepsilon \right) \leq \frac{{\rm Var}(Y)}{\varepsilon^2} \leq \frac{1}{\varepsilon^2 m}.
\end{align}
Let $m= {1}/{\lceil\varepsilon^2 \delta\rceil}$, we have Eq.~\eqref{oest1}.
Now we bound the estimation error of $X_{k_i l_i}$ in step (9) of the direct 
\textbf{Protocol \ref{ptl: original protocol}}. 
We have shown that $\wh{X}_{k_i l_i}$ is an unbiased estimator; 
thus we we use the empirical mean \eqref{eq: wtX k_i l_i} to estimate $X_{k_i l_i}$.
Then we consider the double sum as
\begin{align}
    m \wt{Y} 
    &= \sum_{i=1}^{m} \wt{X}_{k_i l_i} = \sum_{i=1}^{m} \sum_{j=1}^{n_{l_i}} C_{k_i l_i o_j}, 
\end{align}
where 
\begin{align}
    C_{k_i l_i o_j} = \frac{2^n}{n_{l_i} \tr[\psi_{k_i} P_{l_i}]} \times \lambda_{a_j}^{(l_i)} \times \delta_{k_i o_j}.
\end{align}
Note that for a fixed pair $(k_i, l_i)$, the upper bound $b_{k_i l_i}$ and lower bound $a_{k_i l_i}$ of $C_{k_i l_i o_j}$ is
\begin{align}
    b_{k_i l_i} = -a_{k_i l_i} = \frac{\sqrt{2^n}}{n_{l_i} \tr[\psi_{k_i} P_{l_i}]}.
\end{align}
By Hoeffding's inequality \cite{kliesch2021theory}, we have
\begin{align}
  {\rm Pr}(\vert \wt{Y} - Y \vert \geq \varepsilon) 
  &= {\rm Pr}\left(\left\vert \sum_{i=1}^m \sum_{j=1}^{n_{k_i l_i}} C_{k_i l_i o_j} - mY \right\vert \geq m \varepsilon \right) \\
  &\leq 2 \exp\left[-\frac{2 m^2 \varepsilon^2}{\sum_{i=1}^m \sum_{j=1}^{n_{k_i l_i}} \left(\frac{2 \cdot \sqrt{2^n}}{n_{k_i l_i} \tr[\psi_{k_i} P_{l_i}]}\right)^2}\right] \\
  &\leq 2 \exp\left[-\frac{m^2 \varepsilon^2}{\sum_{i=1}^m \sum_{j=1}^{n_{k_i l_i}} \frac{2 \cdot {2^n}}{n_{k_i l_i}^2 \tr[\psi_{k_i} P_{l_i}]^2}}\right] \\
  &\leq \delta,
\end{align}
where $n_{k_i l_i}$ is chosen as
\begin{align}
\label{eq: original protocol n_i}
n_{k_i l_i} &= \left\lceil \frac{1}{\tr[\psi_{k_i} P_{l_i}]^2} \frac{2 \cdot 2^n}{ m \varepsilon^2} \ln{\frac{2}{\delta}}\right\rceil
\end{align}
to guarantee Eq.~\eqref{oest2}.

By combining confidence intervals \eqref{oest1} and \eqref{oest2} using the union bound, 
we establish the confidence interval for the final estimation $\wt{Y}$ 
concluded in Eq.~\eqref{est: eps-close est for original}.

To derive the expected value of the number of calls to the measurement device, 
we note that $n_{kl}$ is a random variable, since $(k, l)$ is randomly chosen. 
By definition, we have
\begin{align}
    \mathbb{E}[n_{k_l}] 
    &= \sum_{i=1}^{\vert \cT \vert} q_{k_i l_i} n_{k_il_i} \\
    & \leq 1+ \sum_{(k_i, l_i)=1}^{\vert \cT \vert} \frac{\tr[\psi_{k_i} P_{l_i}]^2}{2^n} \times \frac{1}{\tr[\psi_{k_i} P_{l_i}]^2} \frac{2 \cdot 2^n}{ m \varepsilon^2} \ln\frac{2}{\delta} \\
    &= 1+ \sum_{(k_i, l_i)=1}^{\vert \cT \vert} \frac{2}{m \varepsilon^2} \ln\frac{2}{\delta} \\
    &= 1 + \frac{2 \vert \cT \vert}{m \varepsilon^2} \ln\frac{2}{\delta},
\end{align}
where $\cT$ is defined in Eq.~\eqref{set T}. Therefore, the expected value of the total number of calls $L = \sum_{i=1}^m n_{k_i l_i}$ is
\begin{align}
  \mathbb{E}[L]
  &= \sum_{i=1}^m \mathbb{E}[n_{k_i l_i}] \\
  &\leq \sum_{i=1}^m \left[1 + \frac{2 \vert \cT \vert}{m \varepsilon^2} \ln\frac{2}{\delta}\right] \\
  &\leq 1 + \frac{1}{\varepsilon^2 \delta} + 2 \vert \cT \vert \frac{1}{\varepsilon^2} \ln\frac{2}{\delta}.
\end{align}
\end{proof}

\section{Proof of Theorem~\ref{thm: magic bounds for original protocol}}
\label{appx: magic bounds for original protocol}

\begin{proof}[Proof of Theorem~\ref{thm: magic bounds for original protocol}]
    Recall the estimator \eqref{eq: ptl o: wtY} in the direct \textbf{Protocol \ref{ptl: original protocol}}. Here, we set $n_{k_i l_i} = 1$ and have
    \begin{align}
        \wt{Y}
        &= \frac{1}{L} \sum_{i=1}^{L} \frac{2^n}{\tr[\psi_{k_i} P_{l_i}]} \times \lambda_{a_i}^{l_i} \times \delta_{k_i o_i} \equiv \frac{1}{m} \sum_{i=1}^{m} C_{k_i l_i o_i},
    \end{align}
    and the upper bound $b_{k_i l_i}$ and the lower bound $a_{k_i l_i}$ of $C_{k_i l_i o_i}$ is $b_{k_i l_i} = - a_{k_i l_i} = \sqrt{2^n}/\tr[\psi_{k_i} P_{l_i}]$. Then by Hoeffding's inequality,  we have
    \begin{align}
      {\rm Pr}(\vert \wt{Y} - Y \vert \geq \varepsilon) 
      &= {\rm Pr}\left(\left\vert \sum_{i=1}^L \frac{2^n}{\tr[\psi_{k_i} P_{l_i}]} \times \lambda_{a_i}^{(l_i)} \times \delta_{k_i o_i} - LY \right\vert \geq L \varepsilon \right) \\
      &\leq 2 \exp\left[-\frac{2 L^2 \varepsilon^2}{\sum_{i=1}^L (b_{k_i l_i} - a_{k_i l_i})^2}\right] \\
      &= 2 \exp\left[-\frac{2 L^2 \varepsilon^2}{\sum_{i=1}^L \frac{4 \cdot 2^n}{\tr[\psi_{k_i} P_{l_i}]^2}}\right] \\ 
      &\leq 2 \exp\left[-\frac{2 L^2 \varepsilon^2}{\sum_{i=1}^L \frac{4 \cdot 2^n}{\mathbb{M}^2}}\right] \\
      &= 2 \exp\left[-\frac{L \varepsilon^2}{2 \cdot 2^n \cdot \mathbb{M}^{-2}}\right] \\
      &\leq \delta,
    \end{align}
    where $\mathbb{M}:= \min_{k, l} \vert \tr[\psi_{k} P_{l}] \vert$. Thus, we have
    \begin{align}
        L
        &\geq \frac{2 \cdot 2^n}{\mathbb{M}^2} \frac{1}{\varepsilon^2} \ln\frac{2}{\delta}.
    \end{align}
    Now, we investigate the relationship between $\mathbb{M}$ and 2-stabilizer Re\'nyi entropy. It is noteworthy that 
    \begin{align}
        \mathbb{M} 
        = \min_{k} \min_{l} \vert \tr[\psi_{k} P_{l}]
        &\leq  \min_{k} \mathbb{E}_l \left[\vert \tr[\psi_{k} P_{l}] \vert\right]  \\
        &\leq  \min_{k} \sqrt{\mathbb{E}_l \left[\tr[\psi_{k} P_{l}]^2\right]}  \\
        &= \min_{k} \sqrt{\frac{\exp[-M_2(\psi_k)]}{2^n}},
    \end{align}
    resulting in the derived lower bound
    \begin{align}
        L
        &\geq \frac{2 \cdot 4^n}{\varepsilon^2} \ln \frac{2}{\delta} \max_k~\exp[M_2(\psi_k)].
    \end{align}

    Now, we delve into the examination of the upper bound. The approach to establish the upper bound involves truncating $\psi_k$ with respect to its Pauli coefficients. This truncated $\psi_k$ is subsequently employed in the execution of the protocol. The specifics are outlined as follows. Initially, we define
    \begin{align}
        \label{def: truncated psi_k}
        \tr\left[\psi_{k,cut} P_{l}\right]
        &:= \left\{
            \begin{array}{lr}  
                 \tr[\psi_k P_{l}],~{\rm if}~\vert \tr[\psi_k P_{l} \vert\geq \frac{\varepsilon}{2\sqrt{2}} \sqrt{\frac{\exp[-M_0(\psi_k)]}{2^n}}, &  \\  
                 0,~{\rm otherwise}. &  
             \end{array} 
        \right.
    \end{align}
    We define $\psi_k':= {\psi_{k,cut}}/{\Vert \psi_{k,cut} \Vert_2}$ where $\Vert \cdot \Vert_2$ denotes Frobenius norm and define
    \begin{align}
        \cT_k^\prime:= \left\{l \Big\vert \vert \tr[\psi_k P_{l} \vert\geq \frac{\varepsilon}{2\sqrt{2}} \sqrt{\frac{\exp[-M_0(\psi_k)]}{2^n}}, \forall l \right\}.
    \end{align}
    We use $\psi_k'$ to replace $\psi_k$, so we have
    \begin{align}
        \cF'
        &= \frac{1}{2^n} \sum_{k=1}^{2^n} \tr[\psi_k'V_k].
    \end{align}
    Now we aim to evaluate $\cF'$ instead of $\cF$. Therefore, we need to bound the difference between $\cF$ and $\cF'$. We have $\left\vert \cF^\prime - \cF \right\vert = \left\vert \sum_{k=1}^{2^n} \tr[\psi_k'V_k] - \sum_{k=1}^{2^n} \tr[\psi_kV_k] \right\vert / 2^n \leq {\sum_{k=1}^{2^n} \left\vert \tr[V_k (\psi_k'- \psi_k)] \right\vert}/{2^n}$ as
    \begin{align}
        \frac{ \sum_{k=1}^{2^n} \left \vert \tr[V_k (\psi_k'- \psi_k)] \right\vert}{2^n}
        &\leq  \sqrt{\frac{\sum_{k=1}^{2^n} \left\vert \tr[V_k (\psi_k' - \psi_k)] \right\vert^2}{2^n}} \\
        &\leq  \sqrt{\frac{\sum_{k=1}^{2^n}  \Vert V_k \Vert_2^2  \Vert \psi_k' - \psi_k \Vert_2^2}{2^n}} \\
        &\leq  \sqrt{\frac{\sum_{k=1}^{2^n} \tr[V_k^2] \cdot \frac{\varepsilon^2}{4}}{2^n}} \\
        &\leq \frac{\varepsilon}{2},
    \end{align}
    where the first inequality is due to AM-QM inequality, the second inequality is due to $\rm H\ddot{o}lder's$ inequality, the third inequality is due to the property that $V_k = V_k^\dagger$ and
    \begin{align}
        \Vert \psi_{k^\prime} - \psi_{k} \Vert_2^2 = 2 - 2 \tr[\psi_{k}^{\prime} \psi_k] = 2 - 2\left[1 - \sum_{l \notin \cT'} \tr[\psi_k P_{l}]^2\right] = 2 - 2\left[1 - \frac{\varepsilon^2}{8}\frac{\vert \overline{\cT_k^\prime} \vert}{\vert \cT_k^\prime \vert}\right] \leq \frac{\varepsilon^2}{4},
    \end{align}
    where $\overline{\cT^\prime}$ is the complement set of the set ${\cT^\prime}$ and we employ Lemma~\ref{lemma: trace(Vk square)} in the final inequality. Consequently, we deduce $\vert \cF' - \cF \vert < \frac{\varepsilon}{2}$. Subsequently, we utilize $\wt{Y}^{\prime}$ to approximate $\cF$ and have the following relation
    \begin{align}
        \vert \cF - \wt{Y}^{\prime} \vert
        = \vert \cF - \cF^{\prime} + F^{\prime} - \wt{Y}^{\prime} \vert
        \leq \vert \cF - \cF' \vert + \vert \cF' - \wt{F}^{\prime} \vert
        < \frac{\varepsilon}{2} + \vert \cF' - \wt{Y}^{\prime} \vert,
    \end{align}
    which implies the sufficient condition $\vert \cF' - \wt{Y}^{\prime} \vert \leq {\varepsilon}/{2}$ to satisfy $\vert \cF - \wt{Y}^{\prime} \vert<\varepsilon$. The expression for $\wt{Y}^{\prime}$ is given by
    \begin{align}
        \wt{Y}^{\prime}
        &= \frac{1}{L} \sum_{i=1}^{L} \frac{2^n}{\tr[\psi_{k_i}^{\prime} P_{l_i}]} \times \lambda_{a_i}^{l_i} \times \delta_{k_i o_i} \equiv \frac{1}{L} \sum_{i=1}^{L} C_{k_i l_i o_i}^{\prime},
    \end{align}
    and the upper bound $b_{k_i l_i}^\prime$ and the lower bound $a_{k_i l_i}^\prime$ of $C_{k_i l_i o_i}^{\prime}$ satisfy $b_{k_i l_i}^\prime = - a_{k_i l_i}^\prime = \sqrt{2^n}/\tr[\psi_{k_i}^{\prime} P_{l_i}]$. By 
 virtue of Hoeffding's inequality, we have
    \begin{align}
        {\rm Pr}(\vert \wt{Y}^{\prime} - \cF' \vert \leq \frac{\varepsilon}{2})
        &= {\rm Pr}\left(\sum_{i=1}^L {C_{k_i l_i o_i}}^{\prime} - L \cF' \leq L \frac{\varepsilon}{2}\right) \\
        &\geq 1 - 2 \exp\left[- \frac{2 \cdot \frac{L^2 \varepsilon^2}{4}}{\sum_{i=1}^L (b_{k_i l_i}^\prime - a_{k_i l_i}^\prime)^2}\right] \\
        &= 1 - 2 \exp\left[- \frac{L^2 \varepsilon^2}{\sum_{i=1}^L \frac{8\cdot 2^n}{\tr[\psi_{k_i}^{\prime} P_{l_i}]^2}}\right] \\
        &\geq 1 - 2 \exp\left[- \frac{L^2 \varepsilon^2}{\sum_{i=1}^L \frac{8 \cdot 2^n}{\mathbb{M}'^2}}\right] \\
        &= 1 - \delta,
    \end{align}
    where
    \begin{align}
        \mathbb{M}^{\prime}
        := \min_{k} \min_{l} \vert \tr[\psi_{k}^{\prime} P_{l}] \vert
        = \min_{k} \min_{l} \frac{\vert\tr[\psi_{k,cut} P_l]\vert}{\Vert \psi_{k,cut} \Vert_2}
        &\geq \min_{k} \min_{l} {\vert\tr[\psi_{k,cut} P_l]\vert} \\
        &\geq \min_k \frac{\varepsilon}{2\sqrt{2}} \sqrt{\frac{\exp[-M_0(\psi_k)]}{2^n}}.
    \end{align}
    Finally, we obtain the upper bound as
    \begin{align}
        L = \frac{8 \cdot 2^n}{\mathbb{M}^{\prime 2}} \frac{1}{\varepsilon^2}\ln\frac{2}{\delta}
        \leq \frac{64 \cdot 4^n}{\varepsilon^4} \ln\frac{2}{\delta} \max_k~\exp[M_0(\psi_k)].
    \end{align}

Substituting $\varepsilon$ and $\delta$ with $2 \varepsilon$ and $2 \delta$, respectively, yields Theorem \ref{thm: magic bounds for original protocol}, presenting the bounding results for a $2 \varepsilon$-close estimator.
\end{proof}

\section{Proof of Proposition~\ref{coro: var comp}}
\label{appx: prop. 8}
\begin{proof}[Proof of Proposition~\ref{coro: var comp}]
    We conclude from the proof of Theorem~\ref{thm: sc of original protocol},
    which concerns the sample complexity of \textbf{Protocol~\ref{ptl: original protocol}}, 
    that
    \begin{align}\label{eq:tmp1}
        c_3 =  \frac{1}{2^n}\sum_k \tr[V_k^2] -\cF^2.
    \end{align}
    In contrast, for \textbf{Protocol~\ref{ptl: efficient protocol}} it holds that
    \begin{align}\label{eq:tmp1}
        c_1 =  \frac{\sum_{l=1}^{4^n} \left[\sum_{k=1}^{2^n}\tr[V_kP_l]\tr[\psi_kP_l]\right]^2}{2^n\left(\sum_{k'=1}^{2^n}\tr[\psi_{k’}P_l]^2\right)} - \cF^2,
    \end{align}
    as evident from the proof of Theorem~\ref{thm: sc of efficient protocol}. 
    Then Cauchy-Schwarz inequality ensures that $c_1 \leq c_3$. 
\end{proof}

\section{More convergence analysis results}
\label{appx: more convergence analysis results}

In this section, we present the convergence analysis outcomes of the EJM schemes 
with $\theta = \pi / 2$ and $\theta = \pi / 4$. 
The visual representations of these results are provided in 
Figs.~\ref{fig: comparisons be two protocol on EJM with 0.5pi} 
and~\ref{fig: comparisons be two protocol on EJM with 0.25pi}, respectively. 
When these findings are considered alongside the results corresponding to $\theta = \pi / 3$, 
as depicted in Fig.~\ref{fig: comparisons be two protocol on EJM with 0.33pi}, 
it becomes evident that our efficient protocol exhibits a advantage over the direct protocol, 
particularly as $\theta$ increases.

\begin{figure}[!hbtp]
	\centering
\includegraphics[width=0.8\textwidth]{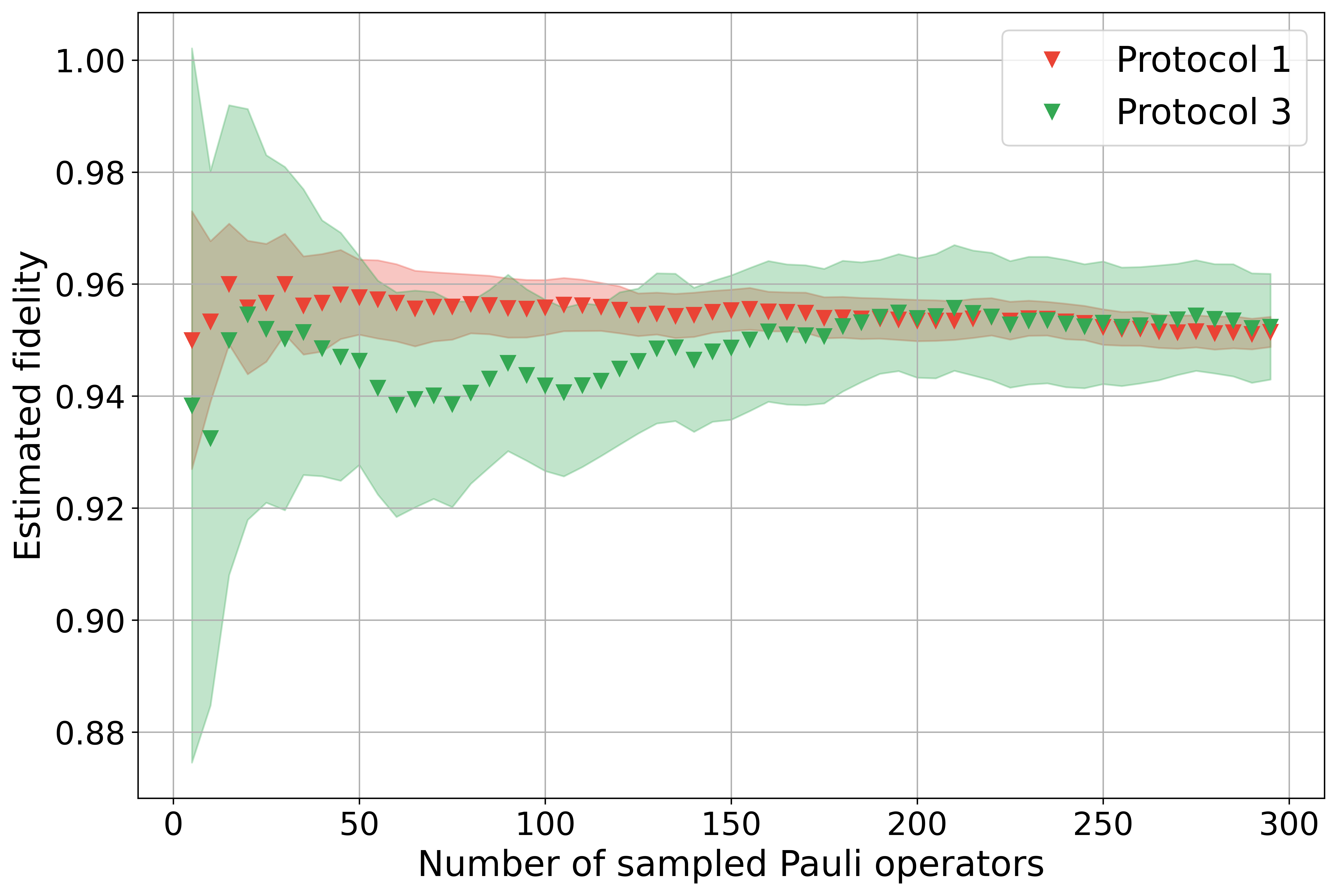}
	\caption{Performance comparisons of two protocols on the EJM scheme with $\theta = \pi / 2$ 
        corresponds to the Bell measurement up to local unitaries \cite{tavakoli2021bilocal}.}
	\label{fig07}
 \label{fig: comparisons be two protocol on EJM with 0.5pi}
\end{figure}

\begin{figure}[!hbtp]
	\centering
\includegraphics[width=0.8\textwidth]{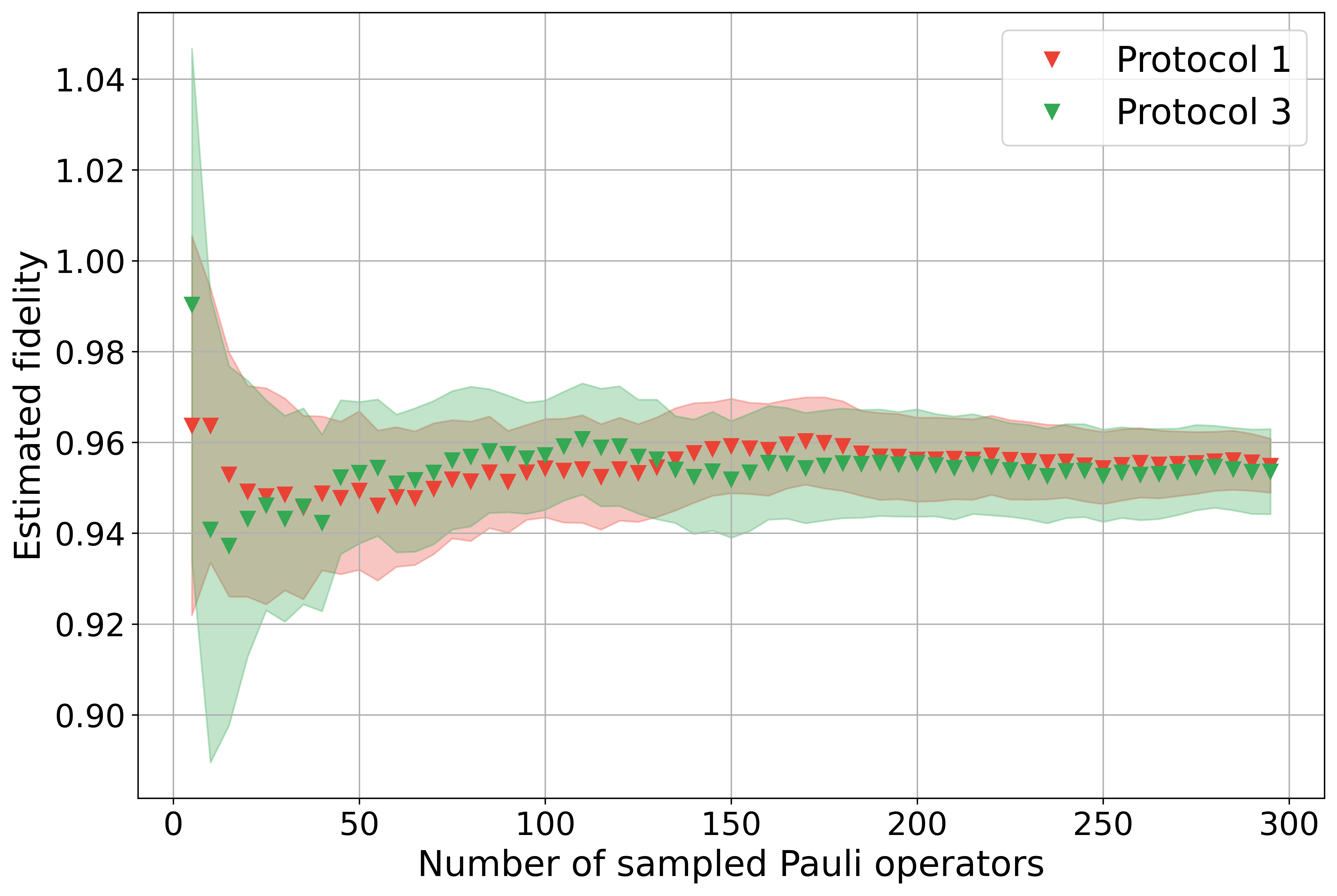}
	\caption{Performance comparisons of two protocols on the EJM scheme with $\theta = \pi / 4$.}
	\label{fig2}
 \label{fig: comparisons be two protocol on EJM with 0.25pi}
\end{figure}

\end{document}